\documentclass[a4paper,UKenglish,cleveref, autoref, thm-restate]{lipics-v2021}

\usepackage{color,xcolor}
\usepackage{tikz}
\usetikzlibrary{arrows.meta}
\usetikzlibrary{patterns,patterns.meta}
\usepackage{tikzmacros}
\usepackage[keep-defaults,Tol]{colorblind}
\usepackage{newunicodechar}
\usepackage{hyphenat}
\usepackage[prependcaption,colorinlistoftodos]{todonotes}

\usepackage{colorpalette,macros,appendix-tools,computational-problems}

\TaskNewPerson{Ken}{Ke}{orange}
\definecolor{marceloColor}{RGB}{26, 192, 19} \TaskNewPerson{Marcelo}{Ma}{marceloColor}
\TaskNewPerson{Nicola}{Ni}{brown}
\TaskNewPerson{Jacob}{Ja}{blue}

\title{Directed disjoint paths remains W[1]-hard on acyclic digraphs without large grid minors}

\titlerunning{Directed disjoint paths remains W[1]-hard on DAGs without large grid minors}

\author{Ken-ichi Kawarabayashi}{National Institute of Informatics, Japan \and \url{https://research.nii.ac.jp/~k_keniti/} }{k_keniti@nii.ac.jp}{https://orcid.org/0000-0001-6056-4287}{}
\author{Nicola Lorenz}{Universität Hamburg, Germany}{nicola.lorenz@uni-hamburg.de}{https://orcid.org/0009-0000-1991-1523}{}
\author{Marcelo {Garlet Milani}}{National Institute of Informatics, Japan \and \url{https://mgarletmilani.com} }{research@mgarletmilani.com}{https://orcid.org/0000-0001-8398-4751}{}
\author{Jacob Stegemann}{Universität Hamburg, Germany}{jacob.stegemann@uni-hamburg.de}{https://orcid.org/0009-0000-6847-1499}{}

\authorrunning{K. Kawarabayashi, N. Lorenz, M. G. Milani, J. Stegemann} 
\Copyright{Ken-ichi Kawarabayashi, Nicola Michelle Lorenz, Marcelo Garlet Milani and Jacob Stegemann} 



\funding{Research of Ken-ichi Kawarabayashi and Marcelo Garlet Milani supported by JSPS KAKENHI JP20A402 and 22H05001, and by JST ASPIRE JPMJAP2302.}

\nolinenumbers

\begin{document}

\maketitle

\begin{abstract}
	In the \kLinkageCongestion/ problem, the input consists of a digraph \(D\), an integer \(c\) and
	\(k\) pairs of vertices \((s_i, t_i)\), and the task
	is to find a set of paths connecting each \(s_i\) to its corresponding \(t_i\),
	whereas each vertex of \(D\) appears in at most \(c\) many paths.
	The case where \(c = 1\) is known to be NP-complete even if \(k = 2\) [Fortune, Hopcroft and Wyllie, 1980] on general digraphs and
	is \W/[1]-hard with respect to \(k\) (excluding the possibility of an \(f(k)n^{O(1)}\)-time algorithm under standard assumptions) on acyclic digraphs [Slivkins, 2010].
	The proof of [Slivkins, 2010] can also be adapted to show \W/[1]-hardness with respect to \(k\) for every congestion \(c \geq 1\).

	We strengthen the existing hardness result by showing that the problem remains \W/[1]-hard for every congestion \(c \geq 1\)
	even if:
	\begin{itemize}
    \item the input digraph \(D\) is acyclic,
    \item \(D\) does not contain an acyclic \((5, 5)\)-grid  as a butterfly minor,
    \item \(D\) does not contain an acyclic tournament on 9 vertices as a butterfly minor, and
    \item \(D\) has ear-anonymity at most 5.
	\end{itemize}

	Further, we also show that the edge-congestion variant of the problem remains \W/[1]-hard for every congestion \(c \geq 1\)
	even if:
	\begin{itemize}
    \item the input digraph \(D\) is acyclic,
    \item \(D\) has maximum undirected degree 3,
    \item \(D\) does not contain an acyclic \((7, 7)\)-wall as a weak immersion and
    \item \(D\) has ear-anonymity at most 5.
	\end{itemize}
\end{abstract}

\newpage
\tableofcontents

\TaskDone{Create todo-list.}

\TaskMarcelo[done]{Write abstract and intro}

\section{Introduction}
\label{sec:introduction}

The \kLinkage/ (\kLinkageShort/) problem is a classic problem for directed graphs (digraphs)
where the input consists of a digraph \(D\) and \(k\) vertex pairs (called terminal pairs)
of the form \((s_i, t_i)\), and the task is to decide
whether \(k\) pairwise disjoint paths \(P_{1}, P_{2}, \ldots, P_{k}\) exist in \(D\)
such that each \(P_i\) is a path from \(s_i\) to \(t_i\).

\kLinkageShort/ is \NP/-complete even if \(k = 2\) \cite{FHW1978}.
On directed acyclic graphs (DAGs), the problem can be solved in \(n^{\Bo{k}}\) time \cite{FHW1978}
and 
is \W/[1]-hard with respect to \(k\) 
(that is, there is no \(f(k) n^{\Bo{1}}\)-time algorithm under standard assumptions) \cite{slivkins2010parameterized}.
Under the assumption of the Exponential Time Hypothesis (ETH),
no \(n^{o(k)}\)-time algorithm exists for \kLinkageShort/ on DAGs \cite{Chitnis23}, implying that
the \(n^{\Bo{k}}\)-time algorithm is essentially optimal.

We investigate the impact on \kLinkageShort/ on DAGs of excluding minors and
imposing further structural restrictions on the input digraph.
There are several notions of minors for digraphs, including
topological minors, butterfly minors and strong minors.
Strong minors are not meaningful on acyclic digraphs,
and if a digraph \(D\) contains \(H\) as a topological minor,
then it also contains \(H\) as a butterfly minor.
Hence, out of the three relations above, butterfly minors are the most adequate for our results.

\emph{Directed tree-width} was introduced by \cite{johnson2001directed},
who also proved that
\kLinkageShort/ can be solved in \(n^{f(k, \DTreewidth{D})}\) time,
where \(\DTreewidth{D}\) is the \emph{directed tree-width} of the input digraph~\(D\).
As the directed tree-width of DAGs is 0, this roughly generalizes 
the \(n^{\Bo{k}}\)-time algorithm for DAGs mentioned above.

The edge-disjoint variant of the problem is called \EdgeDisjointPaths/ (\EdgeDisjointPathsShort/).
\cite{CKK24} proved that \EdgeDisjointPathsShort/ can be solved in \(f(k)n^{\Bo{1}}\)-time on Eulerian digraphs.
\EdgeDisjointPathsShort/ is
\W/[1]-hard with respect to \(k\) even if the input DAG is planar \cite{Chitnis23}.
We observe that the reduction from the edge-disjoint to the vertex-disjoint variant of \EdgeDisjointPathsShort/ mentioned above
does not preserve planarity, and
so this hardness result does not immediately transfer to the vertex-disjoint variant.
In fact, vertex-disjoint \EdgeDisjointPathsShort/ is solvable in \(f(k)n^{\Bo{1}}\)-time on planar digraphs \cite{CyganMPP13}.

Relaxations of \kLinkageShort/ were considered in order to find tractable settings.
One such relaxation is \kLinkageCongestion/ (\kLinkageCongestionShort/),
where each vertex/edge can be used by at most \(c\) many paths in the solution
instead of requiring disjointness.
The value \(c\) is referred to as the maximum \emph{congestion} at a vertex/edge.

The hardness reduction due to \cite{slivkins2010parameterized}
can be adapted to \kLinkageCongestionShort/,
also yielding \W/[1]-hardness with respect to \(k\) on DAGs \cite{AKMR19}.
This was further strengthened by \cite{AKMR19}, who showed that, assuming the \emph{Exponential Time Hypothesis} (ETH),
for every computable function \(f\) and every constant \(c \geq 1\),
there is no \(f(k) n^{o(k / \log{k})}\)-time algorithm for \kLinkageCongestionShort/ on DAGs.

\cite{Wlodarczyk2024} studied an approximation variant of \kLinkageShort/
where the goal is to either find a solution for the given instance or
decide that no solution joining at least \(k / q\) terminal pairs exist, for some constant \(q\).
They proved that such a variant is \W/[1]-hard on DAGs for all constants \(q\).

A maximization variant of \kLinkageShort/ was also studied.
Approximation properties of this variant have been considered by
\cite{ChekuriKS06,Wlodarczyk2024,ChalermsookLN14,ChekuriKS09,ChekuriKS04,ChuzhoyKL16,KolliopoulosS04,ChuzhoyKN22,EneMPR16,KleinbergT98}.

The undirected analog of \kLinkageShort/ is computationally much easier
and can be solved in \(f(k) n^{3}\) time for some function \(f\) \cite{RobertsonS95b}.
The exponent of \(n\) was later improved to 2 \cite{KawarabayashiKR12} and 
recently to \(1 + o(1)\) \cite{KorhonenPS24}.

In rough terms, the algorithms of \cite{RobertsonS95b, KawarabayashiKR12, KorhonenPS24} above work by
applying the \emph{irrelevant vertex technique} in order to repeatedly remove vertices from the graph
until the \emph{tree-width} is sufficiently small so that the problem can be solved efficiently.
Towards this end, the algorithms rely on the \emph{Flat Wall Theorem} \cite{RobertsonS95b},
which roughly states that, for every graph \(G\), we have three cases:
\(G\) contains a large clique as a minor,
\(G\)	contains a large \emph{flat wall}, or
\(G\) has bounded tree-width.

On a somewhat different direction, \cite{Milani24} introduced the parameter
\emph{ear anonymity}, a parameter which tries to capture the ``structure'' of maximal paths in a digraph, and
asked if \kLinkageShort/ can be solved in \(f(k) n^{g(\eanon{D})}\) time on DAGs,
where \(\eanon{D}\) is the \emph{ear anonymity} of the input digraph \(D\).

We prove that \kLinkageCongestionShort/ and its edge-disjoint variant remain computationally hard
even if we impose structural restrictions as described above.
We obtain additional lower bounds when assuming the Exponential Time Hypothesis (ETH).
Below, \(\TT{9}\) is the unique acyclic orientation of the undirected clique on 9 vertices,
and the acyclic \((5, 5)\)-grid is the orientation of the undirected \((5,5)\)-grid where
each edge is oriented from ``left'' to ``right'' and from ``top'' to ``bottom''.

\begin{restatable*}{theorem}{hardnessVertexDisjoint}
    \label{statement:linkage-is-w1-hard-even-with-no-grid-minor}
    For every congestion \(g \geq 1\), \kLinkageCongestion/ is 
		\W/[1]-hard with respect to the number \(k\) of terminal pairs,
		even if
    \begin{itemize}
        \item The input digraph \(D\) is acyclic,
        \item \(D\) contains no \(\TT{9}\) as a butterfly minor,
				\item \(D\) contains no acyclic \((5, 5)\)-grid as butterfly minor, and
        \item \(\eanon{D} \leq 5\).
    \end{itemize}
	Furthermore, assuming the ETH, no \(f(k) n^{o(\sqrt{k/g})}\) time algorithm exists for \kLinkageCongestionShort/,
	even under the above conditions.
\end{restatable*}

Our reduction may create large complete bipartite digraphs where all edges are oriented from one partition to the other.
Hence, it would be interesting to determine whether \kLinkageShort/ remains \W/[1]-hard
even on DAGs without large acyclic grids and without such complete bipartite digraphs as butterfly minors.

For the edge-disjoint variant,
we consider \emph{weak immersions} instead of butterfly minors,
as immersions are more closely related to edge-disjoint than to vertex-disjoint paths.
Since the maximum degree is closed under immersions,
restricting the maximum degree of a digraph \(D\) also exclude the
possibility of \(D\) having any digraph with larger degree as a weak immersion.
For this reason, we consider here the acyclic wall instead of the acyclic grid,
as the grid contains vertices of degree 4.

\newcommand{\HardnessEdgeDisjointStatement}{	\label{statement:linkage-is-w1-hard-even-with-no-wall-immersion}
	For every congestion \(g \geq 1\), \EdgeDisjointPathsCongestion/ is 
	\W/[1]-hard with respect to the number \(k\) of terminal pairs,
	even if
	\begin{itemize}
			\item The input digraph \(D\) is acyclic,
			\item \(D\) has maximum undirected degree 3,
			\item \(D\) contains no acyclic \((7,7)\)-wall as a weak immersion, and
			\item \(\eanon{D} \leq 5\).
	\end{itemize}
	Furthermore, assuming the ETH, no \(f(k) n^{o(\sqrt{k/g})}\) time algorithm exists for \EdgeDisjointPathsCongestion/,
	even under the above conditions.
	}

\begin{restatable*}{theorem}{hardnessEdgeDisjoint}
\HardnessEdgeDisjointStatement
\end{restatable*}

Observe that the hardness reductions of \cite{slivkins2010parameterized,AKMR19}
do not exclude the existence of a \(\TT{k}\) in the reduced instance, and
the reduction due to \cite{Chitnis23} proving for the planar case
does not exclude the existence of acyclic grids of size \(n\).

\TaskMarcelo[done]{Section overview}

The paper is organized as follows.
\Cref{sec:preliminaries} contains definitions and notation.
In \cref{sec:vdp} we prove \cref{statement:linkage-is-w1-hard-even-with-no-grid-minor}.
\Cref{sec:vdp-reduction} contains our hardness reduction,
\cref{sec:vdp-no-grid,sec:vdp-no-tt,sec:vdp-low-ea} contain the structural analysis and
\cref{sec:vdp-hardness} completes the proof.
In \Cref{sec:edp} we prove \cref{statement:linkage-is-w1-hard-even-with-no-wall-immersion}.

\section{Preliminaries}
\label{sec:preliminaries}

\subparagraph{Sets}
For $n \in \N$ we define $[n] \coloneq \Set{1, \dots, n}$.
A \emph{multiset over a set \(X\)} is a function $f \colon X \to \N$.
For every $x \in X$, the value $f(x)$ is the \emph{multiplicity} of that element in the multiset.
We write multisets as
\(
		\Set{x_1 \colon f(x_1), x_2 \colon f(x_2), \ldots, x_n \colon f(x_n)},
\)
since this uniquely defines $f$.
When explicitly stated to be a multiset, we may also write multisets like
\(
		\Set{x_1, x_2, \ldots, x_n}.
\)
In this case, the underlying set $X$ is $\Set{x_1, x_2, \ldots, x_n}$ (understood as a usual set) and the multiplicity of an element $x \in X$ is given by $f(x) \coloneq \Abs{\CondSet{i \in [n]}{x_i = x}}$.

\Task[done]{Fill in definitions and notation.}

\subparagraph{Digraphs} 
We refer the reader to \cite{BG2009} for an introduction to digraph theory.
A directed graph (digraph) \(D\) is a tuple \((V, E)\)
where \(E \subseteq V \times V\) and 
\(E\) does not contain \emph{loops}, that is, edges of the form \((v,v)\).
We write \(\V{D} \coloneq V\) and \(\A{D} \coloneq E\).

Given a set $X \subseteq \V{D}$, we write $D - X$ for the digraph
$(Y \coloneqq \V{D} \setminus X$, $\A{D} \cap Y \times Y)$. 
Similarly, given a set $F \subseteq \A{D}$, we write $D - F$ for the digraph $(\V{D}, \A{D} \setminus F)$.

If $D$ is a digraph and $v \in V(D)$, then 
$\InN[D]{v} \coloneqq \{ u \in V \mid (u,v) \in \A{D}\}$ is the set of
\emph{in-neighbours} and $\OutN[D]{v} \coloneqq \{ u \in V \mid (v,u) \in \A{D}\}$
the set of \emph{out-neighbours} of $v$.
By $\Indeg[D]{v} \coloneqq \Abs{\InN{v}}$ we denote the \emph{in-degree} of
$v$ and by $\Outdeg[D]{v} \coloneqq \Abs{\OutN{v}}$  its out-degree.

Given a vertex \(v\) in a digraph \(D\),
the \emph{split} of \(v\) is the operation which consists of replacing \(v\) with
two vertices \(v_{\text{in}}, v_{\text{out}}\),
adding the edge \((v_{\text{in}}, v_{\text{out}})\) and
the edge \((u, v_{\text{in}}), (v_{\text{out}}, w)\)
for each \(u \in \InN{v}\) and
each \(w \in \OutN{v}\).

\subparagraph{Paths and walks.}
A \emph{walk} of length $\ell$ in a digraph $D$ is a sequence of
vertices \(W \coloneqq \Brace{v_0, v_1, \dots, v_{\ell}}\) such that
$\Brace{v_i, v_{i+1}} \subseteq \A{D}$, for all $0 \leq i < \ell$.
We say that $W$ is a $v_0$-$v_\ell$-walk.

A walk $W \coloneqq \Brace{v_0, v_1, \dots, v_{\ell}}$ is called a
\emph{path} if no vertex appears twice in it and it is called a \emph{cycle} if $v_0 = v_\ell$ and $v_i \neq v_j$ for all $0 \leq i
< j < \ell$. 
We often identify a walk $W$ in $D$ with the corresponding subgraph
and write $V(W)$ and $E(W)$ for the set of vertices and arcs
appearing on it.

Given two walks $W_1 \coloneqq (x_1, x_2, \dots, x_{j})$ and $W_2 \coloneqq (y_1,
y_2, \dots, y_{k})$, we write $W_1 \cdot W_2$
for the \emph{concatenation} of \(W_1\) and \(W_2\), defined as follows.
If $x_j = y_1$, \(W_1 \cdot W_2 \coloneqq (x_1\), \(x_2\), \(\dots,\) \(x_{j}\), \(y_2\), \(y_3\), \(\dots,\) \(y_{k})\).
If $W_1$ or $W_2$ is an empty sequence, then $W_1 \cdot W_2$ is the other walk (or the empty sequence if both walks are empty).
Finally, if \(x_j \neq y_1\), \(W_1 \cdot W_2 \coloneq W_1 \cdot (x_j, y_1) \cdot W_2\).

Given a walk $W \coloneq (x_1, x_2, \dots, x_n)$, we write $W x_i$ for the walk $(x_1, \dots, x_i)$ and similarly $x_i W$ for the walk $(x_i, \dots, x_n)$.

A digraph \(D\) without any cycles is called a \emph{directed acyclic graph} (DAG).
A vertex \(v\) in \(D\) is called a \emph{source} if \(\Indeg[D]{v} = 0\)
and a \emph{sink} if \(\Outdeg[D]{v} = 0\).

Let \(x_{1}, x_{2}, \ldots, x_{n}\) be an ordering of the vertices of a digraph \(D\).
We say that this ordering is a \emph{topological ordering} of \(D\)
if for every edge \((x_i, x_j) \in \E{D}\) we have \(i < j\).
We note that a digraph admits a topological ordering if, and only if, it is a DAG.

\subparagraph{Digraph classes}
An \emph{in-tree with (out-tree) root $r \in V(D)$} is a DAG where
for all $v \in V(D)$ there exists exactly one directed path from $v$ to $r$ (from \(r\) to \(v\) for out-trees).

The \emph{transitive tournament on \(k\) vertices} is the digraph \(\TT{k}\)
with vertex set \(v_{1}, v_{2}, \ldots, v_{k}\) and
edge set \(\CondSet{(v_i, v_j)}{1 \leq i < j \leq k}\).

We consider two types of directed grids (see \cref{fig:grid-and-wall}).
    An \emph{acyclic \((p,q)\)-grid} is the digraph with vertex set
		\(\CondSet{(r,c)}{1 \leq r \leq p \text{ and } 1 \leq c \leq q}\) and
		edge set
		\(\{((r, c), (r + 1, c)) \mid 1 \leq r \leq p - 1 \text{ and } 1 \leq c \leq q\}\) \(\cup\)
		\(\{((r, c), (r, c + 1)) \mid 1 \leq r \leq p \text{ and } 1 \leq c \leq q - 1\}\).
		                                        
				An \emph{acyclic \((p,q)\)-wall} is the digraph obtained from an
		acyclic \((p,q)\)-grid by applying the split operation
		to every vertex of indegree and outdegree exactly 2.

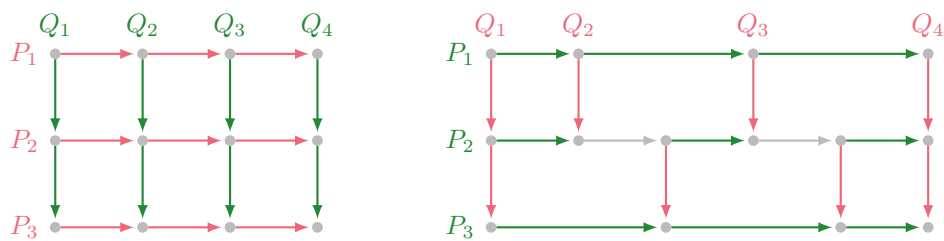
\begin{figure}
	\centering
		\begin{tikzpicture}
			\node[vertex, label = above:{\color{T-Q-B3}{$Q_1$}}, line width = 1.2, fill = T-Q-B0]
	(n0) at (0, 2.3){};
\node[left = 0cm of n0] (l0) {\color{T-Q-B5}{\(P_1\)}};
\node[vertex, label = above:{\color{T-Q-B3}{$Q_2$}}, line width = 1.2, fill = T-Q-B0]
	(n7) at (1.15, 2.3){};
\node[vertex, label = above:{\color{T-Q-B3}{$Q_3$}}, line width = 1.2, fill = T-Q-B0]
	(n8) at (2.3, 2.3){};
\node[vertex, label = above:{\color{T-Q-B3}{$Q_4$}}, line width = 1.2, fill = T-Q-B0]
	(n9) at (3.45, 2.3){};
\node[vertex, label = left:{\color{T-Q-B5}{$P_2$}}, line width = 1.2, fill = T-Q-B0]
	(n10) at (0, 1.15){};
\node[vertex, line width = 1.2, fill = T-Q-B0]
	(n11) at (1.15, 1.15){};
\node[vertex, line width = 1.2, fill = T-Q-B0]
	(n1) at (2.3, 1.15){};
\node[vertex, line width = 1.2, fill = T-Q-B0]
	(n2) at (3.45, 1.15){};
\node[vertex, label = left:{\color{T-Q-B5}{$P_3$}}, line width = 1.2, fill = T-Q-B0]
	(n3) at (0, 0){};
\node[vertex, line width = 1.2, fill = T-Q-B0]
	(n4) at (1.15, 0){};
\node[vertex, line width = 1.2, fill = T-Q-B0]
	(n5) at (2.3, 0){};
\node[vertex, line width = 1.2, fill = T-Q-B0]
	(n6) at (3.45, 0){};
\path[directededge, line width = 0.8, draw = T-Q-B5]
	(n0) to (n7);
\path[directededge, line width = 0.8, draw = T-Q-B5]
	(n7) to (n8);
\path[directededge, line width = 0.8, draw = T-Q-B5]
	(n8) to (n9);
\path[directededge, line width = 0.8, draw = T-Q-B5]
	(n11) to (n1);
\path[directededge, line width = 0.8, draw = T-Q-B5]
	(n1) to (n2);
\path[directededge, line width = 0.8, draw = T-Q-B5]
	(n4) to (n5);
\path[directededge, line width = 0.8, draw = T-Q-B5]
	(n5) to (n6);
\path[directededge, line width = 0.8, draw = T-Q-B5]
	(n10) to (n11);
\path[directededge, line width = 0.8, draw = T-Q-B5]
	(n3) to (n4);
\path[directededge, line width = 0.8, draw = T-Q-B3]
	(n0) to (n10);
\path[directededge, line width = 0.8, draw = T-Q-B3]
	(n10) to (n3);
\path[directededge, line width = 0.8, draw = T-Q-B3]
	(n7) to (n11);
\path[directededge, line width = 0.8, draw = T-Q-B3]
	(n11) to (n4);
\path[directededge, line width = 0.8, draw = T-Q-B3]
	(n8) to (n1);
\path[directededge, line width = 0.8, draw = T-Q-B3]
	(n1) to (n5);
\path[directededge, line width = 0.8, draw = T-Q-B3]
	(n9) to (n2);
\path[directededge, line width = 0.8, draw = T-Q-B3]
	(n2) to (n6);

		\end{tikzpicture}
		\hspace{1cm}
		\begin{tikzpicture}
			\node[vertex, label = left:{\color{T-Q-B3}{$P_1$}}, line width = 1.2, fill = T-Q-B0]
	(n0) at (0, 2.3){};
\node[above = 0cm of n0] (l0) {\color{T-Q-B5}{\(Q_1\)}};
\node[vertex, label = above:{\color{T-Q-B5}{$Q_2$}}, line width = 1.2, fill = T-Q-B0]
	(n9) at (1.15, 2.3){};
\node[vertex, label = above:{\color{T-Q-B5}{$Q_3$}}, line width = 1.2, fill = T-Q-B0]
	(n10) at (3.45, 2.3){};
\node[vertex, label = above:{\color{T-Q-B5}{$Q_4$}}, line width = 1.2, fill = T-Q-B0]
	(n11) at (5.75, 2.3){};
\node[vertex, label = left:{\color{T-Q-B3}{$P_2$}}, line width = 1.2, fill = T-Q-B0]
	(n12) at (0, 1.15){};
\node[vertex, line width = 1.2, fill = T-Q-B0]
	(n13) at (1.15, 1.15){};
\node[vertex, line width = 1.2, fill = T-Q-B0]
	(n1) at (2.3, 1.15){};
\node[vertex, line width = 1.2, fill = T-Q-B0]
	(n2) at (3.45, 1.15){};
\node[vertex, line width = 1.2, fill = T-Q-B0]
	(n3) at (4.6, 1.15){};
\node[vertex, line width = 1.2, fill = T-Q-B0]
	(n4) at (5.75, 1.15){};
\node[vertex, label = left:{\color{T-Q-B3}{$P_3$}}, line width = 1.2, fill = T-Q-B0]
	(n5) at (0, 0){};
\node[vertex, line width = 1.2, fill = T-Q-B0]
	(n6) at (2.3, 0){};
\node[vertex, line width = 1.2, fill = T-Q-B0]
	(n7) at (4.6, 0){};
\node[vertex, line width = 1.2, fill = T-Q-B0]
	(n8) at (5.75, 0){};
\path[directededge, line width = 0.8, draw = T-Q-B3]
	(n0) to (n9);
\path[directededge, line width = 0.8, draw = T-Q-B3]
	(n9) to (n10);
\path[directededge, line width = 0.8, draw = T-Q-B3]
	(n10) to (n11);
\path[directededge, line width = 0.8, draw = T-Q-B3]
	(n6) to (n7);
\path[directededge, line width = 0.8, draw = T-Q-B3]
	(n7) to (n8);
\path[directededge, line width = 0.8, draw = T-Q-B3]
	(n12) to (n13);
\path[directededge, line width = 0.8, draw = T-Q-B3]
	(n5) to (n6);
\path[directededge, line width = 0.8, draw = T-Q-B5]
	(n0) to (n12);
\path[directededge, line width = 0.8, draw = T-Q-B5]
	(n12) to (n5);
\path[directededge, line width = 0.8, draw = T-Q-B5]
	(n9) to (n13);
\path[directededge, line width = 0.8, draw = T-Q-B5]
	(n10) to (n2);
\path[directededge, line width = 0.8, draw = T-Q-B5]
	(n11) to (n4);
\path[directededge, line width = 0.8, draw = T-Q-B5]
	(n4) to (n8);
\path[directededge, line width = 0.8, draw = T-Q-B5]
	(n1) to (n6);
\path[directededge, line width = 0.8, draw = T-Q-B0]
	(n13) to (n1);
\path[directededge, line width = 0.8, draw = T-Q-B3]
	(n1) to (n2);
\path[directededge, line width = 0.8, draw = T-Q-B0]
	(n2) to (n3);
\path[directededge, line width = 0.8, draw = T-Q-B5]
	(n3) to (n7);
\path[directededge, line width = 0.8, draw = T-Q-B3]
	(n3) to (n4);

		\end{tikzpicture}
		\caption{An acyclic \((3,4)\)-grid on the left and an acyclic \((3,4)\)-wall on the right.}
		\label{fig:grid-and-wall}
\end{figure}

\subparagraph{Immersions and minors}

    Let \(D\) and \(H\) be two digraphs.
    A weak immersion of \(H\) in \(D\)
    is a function \(\mu\) with domain \(\V{H} \cup \A{H}\) satisfying the following:
    \begin{enumerate}
        \item \(\mu\) maps \(\V{H}\) injectively into \(\V{D}\),
        \item \(\mu((u, v))\) is a directed path from \(\mu(u)\) to \(\mu(v)\) in \(D\) for every \((u, v) ∈ \A{H}\) and
        \item the directed paths \(\mu(e)\) and \(\mu(f)\) are edge disjoint for every pair of distinct edges \(e, f ∈ \A{H}\).
                    \end{enumerate}

    Let $H$ and $D$ be directed graphs.
    A \textbf{butterfly-model} of \(H\) in \(D\) is a function \(\mu\) which assigns to every \(x \in V(H) \cup \A{H}\) a subdigraph of \(D\)
		such that:
    \begin{enumerate}
        \item for every pair of distinct vertices \(u, v \in V(H)\), \(\mu(u)\) and \(\mu(v)\) are vertex-disjoint,
        \item for a vertex \(v \in V(H)\) and a non-incident edge \(e \in \A{H}\), \(\mu(v)\) and \(\mu(e)\) are vertex-disjoint,
                \item for every \(v \in V(H)\), \(\mu(v)\) is the union of an in-tree and an out-tree intersecting exactly on their common root, and
        \item for every \((u, v) \in \A{H}\), \(\mu((u, v))\) is a directed path starting at a vertex of the out-tree of \(\mu(u)\) and ending at a vertex of the in-tree of \(\mu(v)\).
    \end{enumerate}
    
		Additionally we define a \emph{center} function \(c_{\mu} \colon V(H) \to V(D)\) as follows.
    Let \(v \in V(H)\).
    As \(\mu(v)\) is an acyclic graph, it therefore admits a topological ordering.
    We let \(\Func{c_{\mu}}{v} \in V(\mu(v))\) be minimal in this ordering with the property that it can be reached by all sources of \(\mu(v)\) (note that this does not depend on the choice of the topological ordering).
    
    Note that \(\mu(v)\) being a union of an in-tree and an out-tree intersecting exactly on their common root implies that \(\Func{c_{\mu}}{v}\) can also reach all sinks of \(\mu(v)\). 

\subparagraph{Ear anonymity}

An \emph{ear} in a digraph \(D\) is a subgraph of \(D\) which is either a path or a cycle.
	Let $P$ be an ear.
	A sequence $\Brace{a_1, a_2, \dots, a_{k}}$ of arcs of $P$
	is an \emph{identifying sequence for} $P$
	if \(k \geq 1\) and
	every ear $Q$ containing
	$\Brace{a_1, a_2, \dots, a_{k}}$ in this order
	is a subgraph of $P$.

	Let $P = \Brace{a_1, a_2, \ldots, a_{k}}$ be a maximal ear in a digraph $D$, given by its arc-sequence (in the case of a cycle, any arc of $P$ can be chosen as $a_1$).
	The \emph{ear anonymity} of $P$ in $D$, denoted by $\Anon[D]{P}$, is the length of the shortest identifying sequence for $P$.
	If $k = 0$, we say that $\Anon[D]{P} = 0$.
	The ear anonymity of a digraph $D$, denoted by $\Anon{D}$, is the maximum ear anonymity of the maximal ears of $D$.	

\subparagraph{Complexity theory}
We refer the reader to \cite{CFKLMPPS15,DowneyF13}
for formal definitions of parameterized complexity concepts.

We skip a formal definition of \W/[1]-\emph{hardness} and
recall only the properties that we require for our results.
If a parameterized problem \(L_1\) is \W/[1]-hard with respect to \(k_1\)
and there is a reduction from \(L_1\) to \(L_2\) parameterized by \(k_2\) with running time at most \(f(k_1)n^{O(1)}\),
then \(L_2\) is \W/[1]-hard with respect to \(k_2\).
Under standard assumptions, 
if a problem is \W/[1]-hard with respect to some parameter \(k\),
then no algorithm with running time \(f(k)n^{O(1)}\) exists for said problem.

The \emph{Exponential Time Hypothesis} (ETH) states that \kSat{3} on \(n\) variables and \(m\) clauses
cannot be solved in \(2^{o(n)} \cdot (n + m)^{O(1)}\) time \cite{ImpagliazzoP2001,ImpagliazzoPZ2001}.

\subparagraph{Computational problems}
We recall the definitions of the following decision problems.

\kLinkageCongestionDef

\EdgeDisjointPathsCongestionDef

The special cases where the congestion \(g\) above is 1 are called
\kLinkage/ and \EdgeDisjointPaths/.
Our main reduction is from the following \W/[1]-hard problem.

\GridTilingDef

\begin{theorem}[{\cite[Lemma 1]{Marx12}}]
	\label{stat:grid-tiling-hard}
  \GridTiling/ is \W/[1]-hard with respect to \(k\).
	Furthermore, 
	unless the ETH fails,
	there is no \(f(k) n^{o(k)}\) time algorithm for \GridTiling/.
\end{theorem}

\section{Vertex disjoint paths}
\label{sec:vdp}

We provide a reduction from \GridTiling/ to \kLinkageCongestionShort/.
Our gadgets are inspired on the reduction due to \cite{slivkins2010parameterized}.

\subsection{Hardness reduction}
\label{sec:vdp-reduction}

\TaskMarcelo[done]{Update \cref{const:grid-tiling-no-grid} to allow congestion.}
\begin{construction}
	      \label{const:grid-tiling-no-grid}
		Let \((n, \{S_{i, j}\}_{i, j = 1}^k)\) be an instance of \GridTiling/.
		Let \(g \geq 1\) be an integer.
    We construct a \kLinkageCongestion/ instance \((D, T, g)\) as follows.

    We use two types of \emph{gadgets} in our construction: selector and verifier gadgets.
        For each row and each column, we have one selector gadget \(\SEL{r}{i}\) and \(\SEL{c}{i}\), respectively.
    For each \(\circ \in \Set{r,c}\) and each \(i \in [k]\),
    we construct \(\SEL{\circ}{i}\) as follows (see \cref{fig:sel}).
    
    \begin{figure}[htbp]
        \centering
        \includegraphics[width=0.9\textwidth]{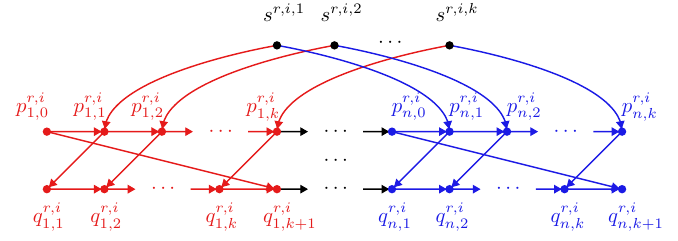}
        \caption{The selector gadget \(\SEL{r}{i}\). The colored subgraphs are repeating building blocks.}
        \label{fig:sel}
    \end{figure}

    For each \(j \in [k]\),
    add the source
    \(s^{\circ, i, j}\) if \(\circ = r\) and
    \(s^{\circ, j, i}\) if \(\circ = c\) instead.

    Define two paths \(P^{\circ, i}\) and \(Q^{\circ, i}\) by their subpaths as follows.
    For each \(1 \leq \ell \leq n\), 
    define
    \begin{align*}
        P^{\circ, i}_\ell = p^{\circ, i}_{\ell, 0}, \ldots, p^{\circ, i}_{\ell, k}
        && \text{ and }
        &&
        Q^{\circ, i}_\ell = q^{\circ, i}_{\ell, 1}, \ldots, q^{\circ, i}_{\ell, k + 1}.
    \end{align*}
    Then, 
    \begin{align*}
        P^{\circ, i} = P^{\circ, i}_1 \cdot \ldots \cdot P^{\circ, i}_n
        && \text{ and }
        && 
        Q^{\circ, i} = Q^{\circ, i}_1 \cdot \ldots \cdot Q^{\circ, i}_n.
    \end{align*}
    Connect the paths above as follows.
		For each \(\ell \in [n]\) and each \(j \in [k]\), 
    add the edges \((p^{\circ, i}_{\ell, 0}, q^{\circ, i}_{\ell, k+1})\) and $(p^{\circ, i}_{\ell, j}, q^{\circ, i}_{\ell, j})$.

    Add the terminal pairs
    \((p^{\circ, i}_{1,0}, p^{\circ, i}_{n, k})\) and
    \((q^{\circ, i}_{1,1}, q^{\circ, i}_{n, k+1})\)
    with multiplicity \(g - 1\) to \(T\).
    Add the terminal pair
    \((p^{\circ, i}_{1,0}, q^{\circ, i}_{n, k + 1})\)
    with multiplicity \(1\) to \(T\).

    This concludes the construction of the selector gadget.
    Vertices and edges of the selector can be summarized as follows.
		\begin{align*}
					V(\SEL{\circ}{i})
							&=
									V(P^{\circ, i}) \cup V(Q^{\circ, i}) \cup \CondSet{s^{\circ, i, j}}{j \in [k]}, \\
					\A{\SEL{\circ}{i}}
							&=
									\A{P^{\circ, i}} \cup \A{Q^{\circ, i}} \\
							&\cup \CondSet{(s^{\circ, i, j}, p^{\circ, i}_{\ell, j}), (p^{\circ, i}_{\ell, j}, q^{\circ, i}_{\ell, j})}{j \in [k], \ell \in [n]} \\
							&\cup \CondSet{(p^{\circ, i}_{\ell, 0}, q^{\circ, i}_{\ell, k+1})}{\ell \in [n]}.
		\end{align*}
    																										        
    Next, we construct the verifier gadgets \(\VER{i}{j}\) for \(1 \leq i, j \leq k\) (see \cref{fig:ver}).
    
    \begin{figure}[htbp]
        \centering
        \includegraphics[width=0.9\textwidth]{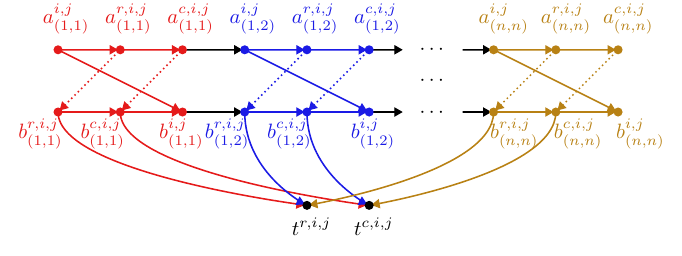}
        \caption{The verifier gadget \(\VER{i}{j}\). The colored subgraphs are repeating building blocks, where the existence of the dashed arrows depends on the set \(S_{i, j}\).}
        \label{fig:ver}
    \end{figure}
    
    First, add the targets \(t^{r, i, j}, t^{c, i, j}\).

    Define two paths \(A^{i, j}\) and \(B^{i,j}\) by their subpaths as follows.
    For \(d \in [n] \times [n]\), let
    \begin{align*}
        A^{i, j}_d = a^{i, j}_d \cdot a^{r, i, j}_d \cdot a^{c, i, j}_d
				&& \text{ and } &&
        B^{i, j}_d = b^{r, i, j}_d \cdot b^{c, i, j}_d \cdot b^{i, j}_d.
    \end{align*}
    For \(1 \leq x \leq n\), let
		\begin{align*}
        A^{i, j}_x = A^{i, j}_{(x, 1)} \cdot \ldots \cdot A^{i, j}_{(x, n)}
				&& \text{and} &&
        B^{i, j}_x = B^{i, j}_{(x, 1)} \cdot \ldots \cdot B^{i, j}_{(x, n)}.
		\end{align*}
    Now define
		\begin{align*}
        A^{i, j} = A^{i, j}_1 \cdot \ldots \cdot A^{i, j}_n
				&& \text{and} &&
        B^{i, j} = B^{i, j}_1 \cdot \ldots \cdot B^{i, j}_n.
		\end{align*}

    Connect \(A^{i,j}\) to \(B^{i,j}\) and \(B^{i,j}\) to \(t^{r,i,j}, t^{c,i,j}\) by adding,
    for each \(d \in S_{i,j}\), 
    the edges
    \((a^{i, j}_d, b^{i, j}_d)\),
    \((a^{r, i, j}_d, b^{r, i, j}_d)\),
    \((a^{c, i, j}_d, b^{c, i, j}_d)\),
    \((b^{r, i, j}_d, t^{r, i, j})\) and
    \((b^{c, i, j}_d, t^{c, i, j})\).

		Add the terminal pairs
		\((a^{i,j}_{(1,1)}, a^{c,i,j}_{(n,n)})\) and
		\((b^{r,i,j}_{(1,1)}, b^{i,j}_{(n,n)})\)
		with multiplicity \(g - 1\) to \(T\).
		Add the terminal pair \((a^{i,j}_{(1,1)}, b^{i,j}_{(n,n)})\)
		with multiplicity 1 to \(T\).
        
    This concludes the construction of the verifier gadgets.
    We can summarize the vertices and edges of the verifier as follows.
    				\begin{align*}
        V(\VER{i}{j})
            &=
                V(A^{i, j}) \cup V(B^{i, j}) \cup \Set{t^{r, i, j}, t^{c, i, j}}.\\
        \A{\VER{i}{j}}
            &=
                \A{A^{i, j}} \cup \A{B^{i, j}} \\
        &\cup \CondSet{\begin{array}{rcl}
             (a^{i, j}_d, b^{i, j}_d), & (a^{r, i, j}_d, b^{r, i, j}_d), & (a^{c, i, j}_d, b^{c, i, j}_d), \\
             & (b^{r, i, j}_d, t^{r, i, j}), &(b^{c, i, j}_d, t^{c, i, j})
        \end{array}}{d \in S_{i, j}}.
		\end{align*}
                                                            														            
    Finally, connect the selectors and verifiers via additional edges as follows.
    \[
        E' = \CondSet{(q^{r, i}_{\ell, j}, a^{r, i, j}_{(\ell, x)}), (q^{c, j}_{\ell, i}, a^{c, i, j}_{(x, \ell)})}{i, j \in [k] \text{ and } \ell, x \in [n]}
    \]

		This completes the construction of \(D\).
		The multiset of terminal pairs is given by
		\begin{align*}
			T & = 
				\CondSet{(s^{r, i, j}, t^{r, i, j}) : 1, (s^{c, i, j}, t^{c, i, j}) : 1, (a^{i, j}_{(1,1)}, b^{i, j}_{n,n}) : 1}{i, j \in [k]} \\
				& \cup \CondSet{(p^{r, i}_{1, 0}, q^{r, i}_{n, k + 1}) : 1, (p^{c, i}_{1,0}, q^{c, i}_{n, k + 1}) : 1}{i \in [k]} \\
				& \cup \CondSet{
					\begin{array}{rl}
						(p^{r,i}_{1, 0}, p^{r,i}_{n, k}) : g - 1,
						& (q^{r,i}_{1,1}, q^{r,i}_{n, k + 1}) : g - 1
						\\
						(p^{c,i}_{1, 0}, p^{c,i}_{n, k}) : g - 1,
						& (q^{c,i}_{1,1}, q^{c,i}_{n, k + 1}) : g - 1
					\end{array}}
				{i \in [k]}\\
				& \cup \CondSet{(a^{i,j}_{(1,1)}, a^{c,i,j}_{(n,n)}) : g - 1, (b^{r,i,j}_{(1,1)}, b^{i,j}_{(n,n)}) : g - 1}{i,j \in [k]}
		\end{align*}
																																		
    \begin{figure}[htbp]
        \centering
        \includegraphics[width=0.9\textwidth]{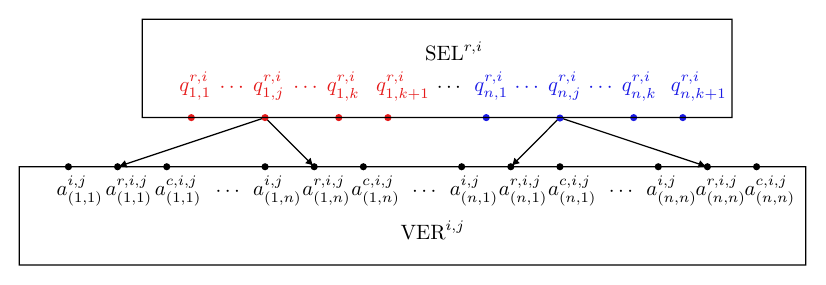}
        \caption{The edges between \(\SEL{r}{i}\) and \(\VER{i}{j}\). The colors in the selector again mark the repeating building blocks.}
        \label{fig:sel-ver}
    \end{figure}
    
    \begin{figure}[htbp]
        \centering
        \includegraphics[width=0.7\textwidth]{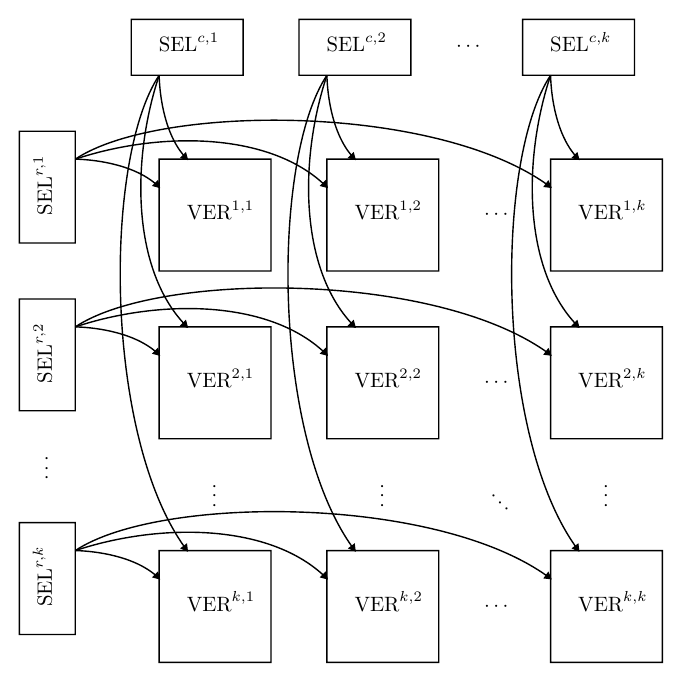}
        \caption{A simplified representation of all edges between selectors and verifiers. Each edge represents in fact multiple edges as shown in \cref{fig:sel-ver}.}
        \label{fig:grid}
    \end{figure}

    In total, \(D\) contains \(k\) selectors \(\SEL{r}{i}\) and
    \(k\) selectors \(\SEL{c}{i}\).
    Each selector contains \(k + 2g - 1\) sources.
    There are \(k^2\) verifier gadgets, each containing \(2g - 1\) sources.
    Hence, we have in total \(2k (k + 2g - 1) + k^2 (2g - 1) = 2g(k^2 + k) + k^2 - 2k\) terminal pairs in the constructed instance.
\end{construction}

Before analyzing the structural properties of the digraph constructed above,
we show that our reduction is correct.

\TaskMarcelo[done]{Update \cref{stat:hardness-correct} for \kLinkageCongestion/.}
\begin{lemma}
	\label{stat:hardness-correct}
	Let \(I_G \coloneqq (n, \{S_{i, j}\}_{i, j = 1}^k)\) be an instance of \GridTiling/.
	Let \(g \geq 1\) be an integer.
	Let \(I_L \coloneqq (D, T, g)\) be the \kLinkageCongestionShort/ instance obtained from \cref{const:grid-tiling-no-grid}.
	Then \(I_G\) is a \emph{yes} instance if, and only if, \(I_L\) is a \emph{yes} instance.
\end{lemma}
\begin{proof}
    
    \textbf{First direction:} If \(I_G\) is a \emph{yes} instance, then \(I_L\) is a \emph{yes} instance.

    Let \(a_1, a_2, \ldots a_k, b_1, b_2, \ldots, b_k \in [n]\) be a solution for \(I_G\).
    We construct a solution for \(I_L\) as follows.
		Recall from \cref{const:grid-tiling-no-grid} that the set of terminal pairs is given by
		\begin{align*}
				T^* & = 
				 \CondSet{
					\begin{array}{rl}
						(p^{r,i}_{1, 0}, p^{r,i}_{n, k}) : g - 1,
						& (q^{r,i}_{1,1}, q^{r,i}_{n, k + 1}) : g - 1
						\\
						(p^{c,i}_{1, 0}, p^{c,i}_{n, k}) : g - 1,
						& (q^{c,i}_{1,1}, q^{c,i}_{n, k + 1}) : g - 1
					\end{array}}
				{i \in [k]}\\
				& \cup  \CondSet{(a^{i,j}_{(1,1)}, a^{c,i,j}_{(n,n)}) : g - 1, (b^{r,i,j}_{(1,1)}, b^{i,j}_{(n,n)}) : g - 1}{i,j \in [k]},\\
			T & = 
				 \CondSet{(s^{r, i, j}, t^{r, i, j}) : 1, (s^{c, i, j}, t^{c, i, j}) : 1, (a^{i, j}_{(1,1)}, b^{i, j}_{n,n}) : 1}{i, j \in [k]} \\
				& \cup  \CondSet{(p^{r, i}_{1, 0}, q^{r, i}_{n, k + 1}) : 1, (p^{c, i}_{1,0}, q^{c, i}_{n, k + 1}) : 1}{i \in [k]}
				\cup  T^*.
				\end{align*}
																																						
		For each \(i, j \in [k]\) and each \(h \in [g - 1]\), we define paths connecting the terminal pairs.

		First, we consider the terminal pairs whose corresponding path does not depend on the solution for \(I_G\).
		
		For each \(\circ \in \Set{r, c}\), the pair \((p^{\circ,i}_{(0,1)}, p^{\circ,i}_{(n,k)})\) gets connected via \(g - 1\) copies of the path 
		\((p^{\circ,i}_{(1,0)}, p^{\circ,i}_{(1,1)}) \cdot P^{\circ,i} \cdot (p^{\circ,i}_{(n, k - 1)}, p^{\circ,i}_{(n, k)})\) and
		the pair \((q^{\circ,i}_{(1,1)}, q^{\circ,i}_{(n, k + 1)})\) gets connected via \(g - 1\) copies of the path 
		\((q^{\circ,i}_{(1,1)}, q^{\circ,i}_{(1,2)}) \cdot Q^{\circ,i} \cdot (q^{\circ,i}_{(n,k)}, q^{\circ,i}_{(n,k + 1)}).\)
		
		The pair \((a^{i,j}_{(1,1)}, a^{i,j,c}_{(n,n)})\) gets connected via \(g - 1\) copies of the path 
		\((a^{i,j}_{(1,1)}, a^{r,i,j}_{(1,1)}) \cdot A^{i,j} \cdot (a^{r,i,j}_{(n,n)}, a^{c,i,j}_{(n,n)})\) and
		the pair \((b^{i,j}_{(1,1)}, b^{i,j}_{(n,n)})\) gets connected via \(g - 1\) copies of the path 
		\((b^{r,i,j}_{(1,1)}, b^{c,i,j}_{(1,1)}) \cdot B^{i,j} \cdot (b^{c,i,j}_{(n,n)}, b^{i,j}_{(n,n)}).\)

		Let \(\mathcal{L}^*\) be the linkage connecting the terminal pairs above.
		Observe that \(\mathcal{L}^*\) provokes a congestion of \(g - 1\) on
		every vertex in \(P^{r,i}, P^{c,i}, Q^{r,i}, Q^{c,i}, A^{i,j}\) and \(B^{i,j}\).

		The solution for the remaining pairs will depend on the values \(a_i\) and \(b_j\) above.
		Recall that the remaining terminal pairs all have multiplicity 1 in \(T\).

    The terminal pair \((s^{r, i, j}, t^{r, i, j})\) gets connected via the path
    \[
        L^{r, i, j} \coloneq (s^{r, i, j}, p^{r, i}_{a_i, j}, q^{r, i}_{a_i, j}, a^{r, i, j}_{(a_i, b_j)}, b^{r, i, j}_{(a_i, b_j)}, t^{r, i, j}).
    \]

    The terminal pair \((s^{c, i, j}, t^{c, i, j})\) gets connected via the path
    \[
        L^{c, i, j} \coloneq (s^{c, i, j}, p^{c, j}_{b_j, i}, q^{c, j}_{b_j, i}, a^{c, i, j}_{(a_i, b_j)}, b^{c, i, j}_{(a_i, b_j)}, t^{c, i, j}).
    \]

    Observe that, because \((a_i, b_j) \in S_{i,j}\), we know that the edges
    \((a^{r, i, j}_{(a_i, b_j)}, b^{r, i, j}_{(a_i, b_j)})\) and
    \((a^{c, i, j}_{(a_i, b_j)}, b^{c, i, j}_{(a_i, b_j)})\) do in fact exist.

    The terminal pair \((a^{i, j}, b^{i, j})\) gets connected via the path
    \[
        L^{i, j} \coloneq (a^{i, j}_{(1, 1)}, a^{r,i, j}_{(1, 1)})
				\cdot A^{i, j}
				\cdot (a^{i, j}_{(a_i, b_j)}, b^{i, j}_{(a_i, b_j)})
				\cdot B^{i, j}
				\cdot (b^{c, i, j}_{(n ,n)}, b^{i, j}_{(n ,n)}).
    \]

    The terminal pair \((p^{r, i}, q^{r, i})\) gets connected via the path
    \[
        L^{r, i} \coloneq (p^{r, i}, p^{r, i}_{1, 0})
				\cdot P^{r, i} 
				\cdot (p^{r, i}_{a_i, 0}, q^{r, i}_{a_i, k+1})
				\cdot Q^{r, i}
				\cdot (q^{r, i}_{n, k+1}, q^{r, i}).
    \]

    The terminal pair \((p^{c, j}, q^{c, j})\) gets connected via the path
    \[
        L^{c, j} \coloneq (p^{c, j}, p^{c, j}_{1, 0})
				\cdot P^{c, j}
				\cdot (p^{c, j}_{b_j, 0}, q^{c, j}_{b_j, k+1})
				\cdot Q^{c, j}
				\cdot (q^{c, j}_{n, k+1}, q^{c, j}).
    \]

    These paths are pairwise vertex disjoint.
                Further, each path visits at most one vertex in each \(P^{r,i}, P^{c,i}, Q^{r,i}, Q^{c,i}, A^{i,j}\) and \(B^{i,j}\).
    Hence, together with the linkage \(\mathcal{L}^*\), we obtain congestion at most \(g\) in order to
    connect all terminal pairs of \(T\).

    \textbf{Second direction:} If \(I_L\) is a \emph{yes} instance, then \(I_G\) is a \emph{yes} instance.

    Let \(\mathcal{L}\) be a solution for \(I_L\).
		For each \(i, j \in [k]\),
		let
    \(L^{r, i, j}\),
    \(L^{c, i, j}\),
    \(L^{r, i}\),
    \(L^{c, j}\) and
    \(L^{i, j}\) 
    be the paths connecting the pairs
    \((s^{r, i, j}, t^{r, i, j})\),
    \((s^{c, i, j}, t^{c, i, j})\),
    \((p^{r, i}_{0,1}, q^{r, i}_{n, k + 1})\),
    \((p^{c, j}_{0,1}, q^{c, j}_{n, k + 1})\) and
    \((a^{i, j}_{(1,1)}, b^{i, j}_{(1,1)})\), respectively, and
    let \(\mathcal{L}' \subseteq \mathcal{L}\) be the set containing these paths.
    Let \(\mathcal{L}^* \subseteq \mathcal{L}\) be
    the multiset of paths connecting
    the pairs in \(T^*\).
    
    Observe that, for each \(L \in \mathcal{L}^*\), the path \(L\) is the unique path in \(D\) from
    \(\Start{L}\) to \(\End{L}\),
    as there are no paths from \(Q^{r,i}\) to \(P^{r,i}\) in the selector gadgets
    and also no path from \(B^{i,j}\) to \(A^{i,j}\) in the verifier gadgets.
    Hence, \(\mathcal{L}^*\) provokes a congestion of \(g - 1\) on
    every vertex in \(P^{r,i}, P^{c,i}, Q^{r,i}, Q^{c,i}, A^{i,j}\) and \(B^{i,j}\).

    As every non-terminal vertex inside \(\mathcal{L}'\) lies in
    \(P^{r,i}, P^{c,i}, Q^{r,i}, Q^{c,i}, A^{i,j}\) or \(B^{i,j}\),
    we deduce that the paths in \(\mathcal{L}'\) are pairwise vertex disjoint,
    as otherwise we would have congestion greater than \(g\) at some vertex.

    Let \(i \in [k]\) and
    consider the path \(L^{r,i} \in \mathcal{L}'\)
    connecting \((p^{r, i}_{0,1}, q^{r, i}_{n, k + 1})\).
    Note that the set of vertices which
    can reach \(q^{r, i}_{n, k + 1}\) and 
    are reachable from \(p^{r,i}_{0,1}\)
    is given by \(\Set{p^{r,i}_{0,1}, q^{r,i}_{n, k + 1}} \cup \V{P^{r,i}} \cup \V{Q^{r,i}}\).
    Hence, \(L^{r,i}\) lies inside \(\SEL{r}{i}\).
    Further, it must contain exactly one edge from \(P^{r,i}\) to \(Q^{r,i}\),
    as there are no edges from \(Q^{r,i}\) back to \(P^{r,i}\).

    Let \(v_p\) be the last vertex along \(L^{r,i}\)
    which lies on \(P^{r,i}\) and
    let \(v_q\) be the first vertex along \(L^{r,i}\)
    which lies on \(Q^{r,i}\).
    Observe that there is
    exactly one path from \(p^{r,i}_{0,1}\) to \(v_p\), which we call \(P'\), and
    exactly one path from \(v_q\) to \(q^{r,i}_{n, k + 1}\), which we call \(Q'\).
    Because of the linkage \(\mathcal{L}^*\),
    the congestion at \(P'\) and
    at \(Q'\) is equal to \(g\), and
    so no other path of \(\mathcal{L}'\) can intersect \(P'\) or \(Q'\).

    Let \(j \in [k]\) and let
    \(v_j\) be the first vertex of \(P^{r,i}\) along \(L^{r,i,j}\).
    By construction, \(v_j\) is necessarily of the form \(p^{r,i}_{a_j, j}\),
    for some \(a_j \in [n]\).
    As \(L^{r,i,j}\) is disjoint from \(P'\),
    the last vertex of \(P'\) must occur before \(p^{r,i}_{a_j,j}\) along \(P^{r,i}\).

    Similarly, let \(u_j\) be the first vertex of \(Q^{r,i}\) along \(L^{r,i,j}\).
    As \(L^{r,i,j}\) is disjoint from \(Q'\),
    the first vertex of \(Q'\) must occur after \(u_j\) along \(Q^{r,i}\).
    Additionally, \(u_j\) must occur at or after \(q^{r,i}_{a_j, j}\) along \(Q^{r,i}\),
    as earlier vertices are not reachable from \(v_j\).

    This implies that \(L^{r,i}\) must contain the edge
    \((p^{r, i}_{a_j, 0}, q^{r, i}_{a_j, k+1})\).
    As this holds for every \(j \in [k]\) and every \(L^{r,i,j}\),
    we have that \(a_j = a_{\ell}\) for all \(j, \ell \in [k]\).
    Let \(a_i \in [n]\) be the value equal to all \(a_j\).
    Also, for each \(j \in [k]\),
    the path \(L^{r,i,j}\) leaves the gadget \(\SEL{r}{i}\)
    at the vertex \(q^{r,i}_{a_i, j}\).

    By an analogous argument, we obtain a value \(b_j\) from the selector
    \(\SEL{c}{j}\) and conclude that,
    for every \(i \in [k]\),
    the path \(L^{c,i,j}\) leaves \(\SEL{c}{j}\) at the vertex
    \(q^{c,i,j}_{b_j, i}\).
    We claim that \(a_{1}, a_{2}, \ldots, a_{k}, b_{1}, b_{2}, \ldots, b_{k}\) is a solution for \(I_G\).

    Let \(i,j \in [k]\).
    We show that \((a_i, b_j) \in S_{i,j}\).
    As shown above, the path \(L^{r,i,j}\) enters
    \(\VER{i}{j}\) by an edge of the form
    \((q^{r,i}_{a_i, j}, a^{r,i,j}_{a_i, x})\), for some \(x \in [n]\).
    Analogously, the path \(L^{c,i,j}\) enters
    \(\VER{i}{j}\) by an edge of the form
    \((q^{c,j}_{b_j, i}, a^{c,i,j}_{\ell, b_j})\), for some \(\ell \in [n]\).

    By an analogous argument as before,
    we conclude that the path \(L^{i,j}\)
    must contain exactly one edge of the form
    \((a^{i,j}_{\ell', x'}, b^{i,j}_{\ell', x'})\),
    as otherwise it would intersect both \(L^{r,i,j}\) and \(L^{c,i,j}\),
    causing congestion greater than \(g\) at some vertex.
    In particular, this implies that \(\ell' = a_i = \ell\) and \(x' = b_j = x\).
    Hence, \(L^{r,i,j}\) must contain the edge
    \((a^{r,i,j}_{a_i, b_j}, b^{r,i,j}_{a_i, b_j})\) and
    \(L^{c,i,j}\) must contain the edge
    \((a^{c,i,j}_{a_i, b_j}, b^{c,i,j}_{a_i, b_j})\).
    This is only possible if \((a_i, b_j) \in S_{i,j}\), as otherwise these edges do not exist in \(D\).

    Hence, for all \(i,j \in [k]\), we have \((a_i, b_j) \in S_{i,j}\), and
    so \(I_G\) is a \emph{yes} instance. \end{proof}

\subsection{Excluding an acyclic \texorpdfstring{\((5, 5)\)}{(5, 5)}-grid}
\label{sec:vdp-no-grid}

\begin{lemma}
    \label{stat:no-grid-in-selector}
    \label{stat:no-grid-in-verifier}
		For each \(1 \leq i,j \leq k\), none of the gadgets
    \(\VER{i}{j}\), \(\SEL{r}{i}\) and \(\SEL{c}{i}\) contain any acyclic \((3, 3)\)-grid as a butterfly minor.
\end{lemma}
	
\begin{proof}
		We prove the statement for \(\SEL{r}{i}\).
		The argument for the other gadgets follow analogously.

    We assume towards a contradiction that \(\SEL{r}{i}\) contains an acyclic \((3, 3)\)-grid \(G\) as a butterfly minor.
    Let \(\mu\) be the model of \(G\) in \(\SEL{r}{i}\).

    Let \(s\) be the vertex of \(G\) with in-degree \(0\).
    The indegree of every \(x \in V(G)\) with \(x \neq s\) is at least \(1\).
    Therefore, \(\Func{c_{\mu}}{x} \neq s^{r, i, j}\) for all \(j \in [k]\), so \(\Func{c_{\mu}}{x} \in V(P^{r, i}) \cup V(Q^{r, i})\).
        
    Next, let \(v \in V(G)\) be the vertex with in- and out-degree \(2\).
    Since \(v \neq s\), we have \(\Func{c_{\mu}}{v} \in V(P^{r, i}) \cup V(Q^{r, i})\).
    
    Assume towards a contradiction that \(\Func{c_{\mu}}{v} \in V(P^{r, i})\).
    Let \(u, w \in V(G)\) be the two in-neighbours of \(v\).
    Then
    \begin{align*}
        \Func{c_{\mu}}{u}, \Func{c_{\mu}}{w}
        &\in \CondSet{x \in V(\SEL{r}{i})}{\text{There exists a path from } x \text{ to } \Func{c_{\mu}}{v}} \\
        &= \CondSet{s^{r, i, j}}{j \in [n]} \cup V(P^{r, i}).
    \end{align*}
    As \(u, w \neq s\), we have \(\Func{c_{\mu}}{u}, \Func{c_{\mu}}{w} \in V(P^{r, i})\).
    Without loss of generality, \(\Func{c_{\mu}}{u}, \Func{c_{\mu}}{w}\) and \(\Func{c_{\mu}}{v}\) appear in this order on \(P^{r, i}\).
    Then every path from \(\Func{c_{\mu}}{u}\) to \(\Func{c_{\mu}}{v}\) contains \(\Func{c_{\mu}}{w}\).
    Thus, \(\mu(w)\) is not disjoint from \(\mu(u) \cup \mu((u, v)) \cup \mu(v)\).
    A contradiction.
    
    Hence, \(\Func{c_{\mu}}{v} \in V(Q^{r, i})\).
    Let \(u', w' \in V(G)\) be the two out-neighbours of \(v\).
    As before, \(\Func{c_{\mu}}{u'}, \Func{c_{\mu}}{w'} \in V(Q^{r, i})\).
		We again get a contradiction as there are no two internally disjoint paths from \(\Func{c_{\mu}}{v}\) to \(\Func{c_{\mu}}{u'}\) and \(\Func{c_{\mu}}{w'}\).
\end{proof}

The following observation can easily be verified from the construction and will be useful in subsequent proofs.

\begin{observation}
    \label{stat:gadget-connectivity}
    Let \((n, \{S_{i, j}\}_{i, j = 1}^k)\) be a \GridTiling/ instance.
		Let \((D, T)\) be the instance given by \cref{const:grid-tiling-no-grid}.
    For all \(i, i', j, j' \in [k]\) and all \(\circ \in \Set{r, c}\), the following holds
    \begin{enumerate}
        \item there are no edges leaving \(V(\VER{i}{j})\),
        \item there are no edges coming into \(V(\SEL{\circ}{i})\),
        \item if there is an edge from \(V(\SEL{r}{i'})\) to \(V(\VER{i}{j})\), then \(i' = i\), and
        \item if there is an edge from \(V(\SEL{c}{j'})\) to \(V(\VER{i}{j})\), then \(j' = j\).
    \end{enumerate}
\end{observation}

\begin{corollary}
    \label{stat:butterfly-minor-restriction}
    Let \((n, \{S_{i, j}\}_{i, j = 1}^k)\) be an instance of \GridTiling/.
    Let \((D, T)\) be the reduced instance constructed by \cref{const:grid-tiling-no-grid}.
    Let \(H\) be an acyclic digraph with exactly one source and exactly one sink.
    If \(H\) is a butterfly minor of \(D\) with model \(\mu\), then there exist \(i, j \in [n]\) such that \(\Image(\mu) \subseteq D[V(\SEL{r}{i}) \cup V(\VER{i}{j})]\) or \(\Image(\mu) \subseteq D[V(\SEL{c}{j}) \cup V(\VER{i}{j})]\).
\end{corollary}

We can now exclude an acyclic \((5, 5)\)-grid by arguing that
such a grid in the constructed digraph would imply
the existence of an \((3, 3)\)-grid in one of the selectors or one of the verifiers.
\begin{lemma}
    \label{stat:no-large-grid}
    Let \((n, \{S_{i, j}\}_{i, j = 1}^k)\) be an instance of \GridTiling/.
		Let \((D, T)\) be the reduced instance constructed by \cref{const:grid-tiling-no-grid}.
    Then \(D\) does not contain any acyclic \((5, 5)\)-grid as a butterfly minor.
\end{lemma}
\begin{proof}
    We assume towards a contradiction that \(D\) contains an acyclic \((5, 5)\)-grid \(G\) as a butterfly minor.
    Let \(\mu\) be the model of \(G\) in \(D\).

    \(G\) has exactly one source and one sink.
    By \cref{stat:butterfly-minor-restriction}, there exist \(i, j \in [n]\) such that \(\Image(\mu) \subseteq D[V(\SEL{r}{i}) \cup V(\VER{i}{j})]\) or \(\Image(\mu) \subseteq D[V(\SEL{c}{j}) \cup V(\VER{i}{j})]\).
    Without loss of generality we assume that \(\Image(\mu) \subseteq D[V(\SEL{r}{i}) \cup V(\VER{i}{j})]\).

    Define a partition of the vertices of \(G\) as follows
    \begin{align*}
			A &\coloneq \CondSet{v \in V(G)}{\Func{c_{\mu}}{v} \in V(\SEL{r}{i})}, \\
      B &\coloneq \CondSet{v \in V(G)}{\Func{c_{\mu}}{v} \in V(\VER{i}{j})}.
    \end{align*}

    By \cref{stat:gadget-connectivity}, there is no edge from \(B\) to \(A\), as otherwise the model would give a path from \(\Func{c_{\mu}}{b} \in V(\VER{i}{j})\) to \(\Func{c_{\mu}}{a} \in V(\SEL{r}{i})\) for some \(b \in B\) and \(a \in A\).
    
    Let \(v\) be the vertex in the center of the acyclic \((5, 5)\)-grid \(G\).

    \begin{figure}[htbp]
        \centering
        \includegraphics[width=0.4\textwidth]{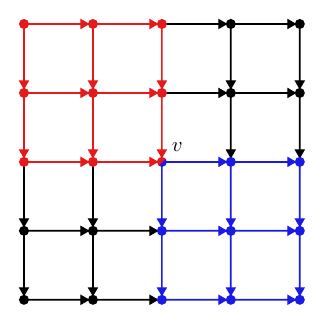}
        \caption{The acyclic \((5, 5)\)-grid \(G\), \(v \in V(G)\) and two acyclic \((3, 3)\)-subgrids (red and blue).}
        \label{fig:5x5}
    \end{figure}
    
    If \(v \in A\), then the top left acyclic \((3, 3)\)-subgrid (see the red colored subgraph in \cref{fig:5x5}) is also in \(A\) as there is no edge from \(B\) to \(A\).
    Since there is no edge from \(\VER{i}{j}\) to \(\SEL{r}{i}\), restricting the domain of \(\mu\) to this acyclic \((3, 3)\)-grid and its codomain to \(\SEL{r}{i}\) would give a model for the acyclic \((3, 3)\)-grid as a butterfly minor in \(\SEL{r}{i}\).
    But by \cref{stat:no-grid-in-verifier}, \(\SEL{r}{i}\) does not contain any acyclic \((3, 3)\)-grid as a butterfly minor.
    Similarly, if \(v \in B\), then the bottom right acyclic \((3, 3)\)-subgrid (see the blue colored subgraph in \cref{fig:5x5}) is also in \(B\).
		We would get an acyclic \((3, 3)\)-grid as a butterfly minor in \(\VER{i}{j}\).
    But by \cref{stat:no-grid-in-selector}, \(\VER{i}{j}\) does not contain any acyclic \((3, 3)\)-grid as a butterfly minor.
    Hence, we obtain a contradiction.
\end{proof}

\subsection{Excluding a \texorpdfstring{\(\TT{9}\)}{K9}}
\label{sec:vdp-no-tt}

\begin{lemma}
    \label{stat:no-tournament-in-selector}
    \label{stat:no-tournament-in-verifier}
    \label{stat:no-tournament-in-gadgets}
		For each \(1 \leq i,j \leq k\), none of the gadgets
    \(\VER{i}{j}\), \(\SEL{r}{i}\) and \(\SEL{c}{i}\)
    contain any \(\TT{6}\) as a butterfly minor.
\end{lemma}
\begin{proof}
	We prove the statement for the selector \(\SEL{r}{i}\).
	The proof for the remaining gadgets follow analogously.
    We assume towards a contradiction that \(\SEL{r}{i}\) contains a \(G \coloneq \TT{6}\) as a butterfly minor.
    Let \(\mu\) be the model of \(G\) in \(\SEL{r}{i}\).
    
    Let \(s\) be the vertex of \(G\) with in-degree \(0\) and \(U \coloneq V(G - s)\).
    Then, the indegree of all vertices \(u \in U\) is at least \(1\).
    Therefore, \(\Func{c_{\mu}}{u} \neq s^{r, i, j}\) for all \(1 \leq j \leq k\).
    
    Define \(A \coloneq V(P^{r, i}), B \coloneq V(Q^{r, i})\).
    The five vertices of \(U\) lie in \(A \cup B\).
    So \(\Abs{U \cap A} \geq 3\) or \(\Abs{U \cap B} \geq 3\).
    But both cases lead to a contradiction as in the end of the proof of \cref{stat:no-grid-in-selector}.
\end{proof}

\TaskMarcelo{Double check \cref{stat:no-tournament}.}

\begin{lemma}
	\label{stat:no-tournament}
	Let \((n, \{S_{i, j}\}_{i, j = 1}^k)\) be an instance of \GridTiling/.
	Let \((D, T)\) be the reduced instance constructed by \cref{const:grid-tiling-no-grid}.
	Then \(D\) contains no \(\TT{9}\) as a butterfly minor.
\end{lemma}

\begin{proof}
    We assume towards a contradiction that \(D\) contains a \(G \coloneq \TT{9}\) as a butterfly minor.
    We let \(v_1, \ldots v_9\) be the vertices of \(G\) in the topological order.
    Let \(\mu\) be the model of $G$ in \(D\).

    It is easy to verify that \(G\) has exactly one source and one sink.
    By \cref{stat:butterfly-minor-restriction}, there exist \(i, j \in [n]\) such that \(\Image(\mu) \subseteq D[V(\SEL{r}{i}) \cup V(\VER{i}{j})]\) or \(\Image(\mu) \subseteq D[V(\SEL{c}{j}) \cup V(\VER{i}{j})]\).
    Without loss of generality we assume that \(\Image(\mu) \subseteq D[V(\SEL{r}{i}) \cup V(\VER{i}{j})]\).
    Define a partition of the vertices of $G$ as follows
    \begin{align*}
        A &\coloneq \CondSet{v \in V(G)}{\Func{c_{\mu}}{v} \in V(\SEL{r}{i})}, \\
        B &\coloneq \CondSet{v \in V(G)}{\Func{c_{\mu}}{v} \in V(\VER{i}{j})}.
    \end{align*}

    There is no edge from \(B\) to \(A\) as otherwise the model would give a path from \(\Func{c_{\mu}}{b} \in V(\VER{i}{j})\) to \(\Func{c_{\mu}}{a} \in V(\SEL{r}{i})\) for some \(b \in B\) and \(a \in A\).
    But this contradicts \cref{stat:gadget-connectivity}.

    Assume towards a contradiction that \(\Abs{A} \geq 6\).
    As there is no edge from \(B\) to \(A\), this implies \(\Func{c_{\mu}}{v_1}, \ldots, \Func{c_{\mu}}{v_6} \in A\).
		Since there is no edge from \(\VER{i}{j}\) to \(\SEL{r}{i}\), restricting \(\mu\) to \(G[\Func{c_{\mu}}{v_1}, \ldots, \Func{c_{\mu}}{v_6}]\) and intersecting it with \(\SEL{r}{i}\) would give a \(\TT{6}\) butterfly minor in \(\SEL{r}{i}\).
    But by \cref{stat:no-tournament-in-verifier}, \(\SEL{r}{i}\) does not contain any \(\TT{6}\) as a butterfly minor.
    Similarly, if \(\Abs{B} \geq 6\) and \(\Func{c_{\mu}}{v_4}, \ldots, \Func{c_{\mu}}{v_9} \in B\),
    we would get a \(\TT{6}\) butterfly minor in \(\VER{i}{j}\).
    But by \cref{stat:no-tournament-in-selector}, \(\SEL{r}{i}\) does not contain any \(\TT{6}\) as a butterfly minor.

    Now, \(\Abs{A} + \Abs{B} = 9\) together with the above statement implies that \(\Abs{A} = 4\) and \(\Abs{B} = 5\) or that \(\Abs{A} = 5\) and \(\Abs{B} = 4\).
    There is no edge from \(B\) to \(A\), so we have \(\Func{c_{\mu}}{v_1}, \dots, \Func{c_{\mu}}{v_4} \in A\) and \(\Func{c_{\mu}}{v_5}, \dots, \Func{c_{\mu}}{v_9} \in B\) or \(\Func{c_{\mu}}{v_1}, \dots, \Func{c_{\mu}}{v_5} \in A\) and \(\Func{c_{\mu}}{v_6}, \dots, \Func{c_{\mu}}{v_9} \in B\).

    In the following we will now show that neither case is true, which proves the statement of the lemma.
    For all \(1 \leq x < y \leq 9\) we let \(R_{x, y}\) be the unique path from
		\(\Func{c_{\mu}}{v_x}\) to \(\Func{c_{\mu}}{v_y}\) satisfying
    \[
        R_{x, y} \subseteq \mu(v_x) \cup \mu((v_x, v_y)) \cup \mu(v_y).
    \]
    Note that for \(1 \leq x < y \leq 9\) and \(1 \leq k \leq 9\) with \(k \neq x, y\) we have \(\Func{c_{\mu}}{v_k} \notin V(R_{x, y})\) and for \(1 \leq k < \ell \leq 9\) with \(\Set{k, \ell} \cap \Set{x, y} = \emptyset\) also \(R_{k, \ell} \cap R_{x, y} = \emptyset\).

		First, we establish a property which will be repeatedly used later.
		\begin{claim}
			\label{claim:no-three-vertices-on-path}
			The sets
			\(A \cap \V{P^{r,i}}\),
			\(A \cap \V{Q^{r,i}}\),
			\(B \cap \V{A^{i,j}}\) and
			\(B \cap \V{B^{i,j}}\)
			contain each at most \(2\) vertices.
		\end{claim}
		\begin{claimproof}
			We prove the statement for the set \(A \cap \V{P^{r,i}}\).
			The other cases follow analogously.

			Assume towards a contradiction that there are three vertices, say, \(u_1, u_2, u_3\),
			appearing in this order along \(P^{r,i}\) in \(A \cap \V{P^{r,i}}\).
			Then, every \(u_1\)-\(u_3\) path contains \(u_2\).
			This, however, contradicts the fact that the corresponding \(u_1\)-\(u_3\) path
			\(R_{a,b}\) defined above does not contain \(u_2\).
		\end{claimproof}

    We will call edges of \(\VER{i}{j}\) of the form \((a^{i, j}_{d}, b^{i, j}_{d})\) for some \(d \in [n] \times [n]\) \emph{jumps}.
    For notational convenience we write
    \[
        T \coloneq \Set{t^{r, i, j}, t^{c, i, j}}.
    \]

		\begin{CaseDistinction}
			\Case{Assume towards a contradiction that \(\Func{c_{\mu}}{v_1}\), \(\ldots,\) \(\Func{c_{\mu}}{v_4} \in A\) and \(\Func{c_{\mu}}{v_5}\), \(\ldots,\) \(\Func{c_{\mu}}{v_9} \in B\).}
    
		By \cref{claim:no-three-vertices-on-path}, no three of \(\Func{c_{\mu}}{v_5}, \ldots, \Func{c_{\mu}}{v_9}\) lie on \(A^{i, j}\) (or \(B^{i, j}\)).
    This implies that \(\Func{c_{\mu}}{v_5}, \Func{c_{\mu}}{v_6} \in V(A^{i, j})\) and
		\(\Func{c_{\mu}}{v_7}, \Func{c_{\mu}}{v_8} \in V(B^{i, j})\),
		which implies \(\Func{c_{\mu}}{v_9} \in T\).
    
    Now, \(R_{5,6} \subseteq A^{i,j}\) and \(R_{7,8} \subseteq B^{i,j}\).
    Consider \(R_{5,8}\).
    To avoid intersecting \(\Func{c_{\mu}}{v_6}\) and \(\Func{c_{\mu}}{v_7}\), it has to contain a jump starting before \(\Func{c_{\mu}}{v_6}\) along \(A^{i,j}\) and ending after \(\Func{c_{\mu}}{v_7}\) along \(B^{i,j}\).

    Let \(z \in V(A^{i,j} \cap R_{4,6})\) be the first vertex along \(A^{i,j}\).
    As we have just shown that \(R_{5,8}\) contains a jump starting before \(\Func{c_{\mu}}{v_6}\) along \(A^{i,j}\) and ending after \(\Func{c_{\mu}}{v_7}\) along \(B^{i,j}\), we have \(\Func{c_{\mu}}{v_6} = a^{r,i,j}_d\) or \(\Func{c_{\mu}}{v_6} = a^{c,i,j}_d\) for some \(d \in [n] \times [n]\).
    In either case \(z = a^{r,i,j}_d\) and for all \(k \in [4]\) the paths \(R_{k,6}\) have \(z\) as their first vertex in \(\VER{i}{j}\).
    But there exists only one vertex \(q \in Q^{r,i}\) which as an edge to \(z\).
    Then in particular \(q \in R_{1,6}\), so we have \(q \notin \mu(v_4)\).

    Let \(z' \in V(A^{i,j} \cap R_{4,8})\) be the first vertex along \(A^{i,j}\).
    Consider \(R_{6,9}\).
    We claim that it does contain a jump starting before \(z'\) along \(A^{i,j}\) and ending after \(\Func{c_{\mu}}{v_8}\) along \(B^{i,j}\).
    Assume towards a contradiction that it does not.
    If now \(R_{4,8}\) also contains no jumps, then it would intersect \(R_{6,9}\).
    But even if \(R_{4,8}\) contains jumps, they are either already contained in \(R_{5,6}\) or intersect \(R_{6,9}\) nevertheless.

    Thus, \(R_{6,9}\) contains a jump starting before \(z'\) along \(A^{i,j}\) and ending after \(\Func{c_{\mu}}{v_8}\) along \(B^{i,j}\).
    Hence, for all \(k \in [4]\) the paths \(R_{k,8}\) have \(z'\) as their first vertex in \(\VER{i}{j}\).
    But there exists only one vertex \(q' \in Q^{r,i}\) which has an edge to \(z'\).
    Then in particular \(q' \in R_{1,8}\), so we have \(q' \notin \mu(v_4)\).

    Consider the paths \(R_{4,6} q\) and \(R_{4,8} q'\). 
    As \(\Func{c_{\mu}}{v_4}\) lies on \(Q^{r,i}\), they have to intersect.
    If \(q \in R_{4,8} q'\), then in particular \(q \in \mu(v_4)\).
    A contradiction.
    On the hand, if \(q' \in R_{4,6} q\), then in particular \(q' \in \mu(v_4)\).
    Again a contradiction.

		\Case{Assume towards a contradiction that \(\Func{c_{\mu}}{v_1}\), \dots, \(\Func{c_{\mu}}{v_5} \in A\) and \(\Func{c_{\mu}}{v_6}\), \dots, \(\Func{c_{\mu}}{v_9} \in B\).}

		By \cref{claim:no-three-vertices-on-path}, no three of \(\Func{c_{\mu}}{v_6}, \ldots, \Func{c_{\mu}}{v_9}\) lie on \(A^{i, j}\) (or \(B^{i, j}\)).
		As \(v_1\) is the only source in \(\TT{9}\),
		exactly one vertex of \(A\) is in the set \(\Set{s^{r,i,1}, s^{r,i,2}, \ldots, s^{r,i,k}}\).
		Thus, \(A \cap \V{P^{r,i}} = \Set{\Func{c_{\mu}}{v_2}, \Func{c_{\mu}}{v_3}}\) and
		\(A \cap \V{Q^{r,i}} = \Set{\Func{c_{\mu}}{v_4}, \Func{c_{\mu}}{v_5}}\).

		For each \(s \in [5]\) and each \(t \in \{6,7,8\}\) let \(R_{s,t}'\)
		be the subpath of \(R_{s,t}\) from its start
		until its first intersection with \(\VER{i}{j}\).

		\begin{claim}
			\label{claim:distinct-arrival-at-A-i-j}
			For each \(t \in \{6,7,8\}\), the paths
			\(R_{4,t}'\) and \(R_{5,t}'\) are vertex disjoint.
		\end{claim}
		\begin{proof}
			We prove the statement for \(R_{4,7}'\) and \(R_{5,7}'\).
			The other cases follows analogously.

			Note that \(\Func{c_{\mu}}{v_4}, \Func{c_{\mu}}{v_5} \in \V{Q^{r,i}}\).
			If \(R_{4,7}'\) intersects \(R_{5,7}'\) inside \(\SEL{r}{i}\),
			then \(\Func{c_{\mu}}{v_5} \in \V{R_{4,7}'}\), a contradiction.

			Hence, there are indices \(e < f\)
			such that
			the last vertex of \(R_{4,7}'\) in \(\SEL{r}{i}\) is \(q^{r,i}_{e,j}\),
			the last vertex of \(R_{5,7}'\) in \(\SEL{r}{i}\) is \(q^{r,i}_{f,j}\)
			This implies that
			the first vertex of \(R_{4,7}'\) in \(\VER{i}{j}\)
			is some \(a^{r,i,j}_{e,x_1}\), for some \(x_1\) and 
			the first vertex of \(R_{5,7}'\) in \(\VER{i}{j}\)
			is some \(a^{r,i,j}_{f,x_2}\), for some \(x_2\).
			As \(e < f\), we have
			\(a^{r,i,j}_{e,x_1} \neq a^{r,i,j}_{f,x_2}\).
			Hence, \(R_{4,7}'\) and \(R_{5,7}'\) are vertex disjoint.
		\end{proof}

		\begin{Subcase}
			\Case{\(\Func{c_{\mu}}{v_6}, \Func{c_{\mu}}{v_7}\) lie in \(A^{i,j}\).}

			In particular, there are \(e < h\) and \(x_1, x_4\)
			such that
			\(\Func{c_{\mu}}{v_6} \in \{a^{i,j}_{e,x_1}, a^{r,i,j}_{e,x_1}, a^{c,i,j}_{e,x_1}\}\) and
			\(\Func{c_{\mu}}{v_7} \in \{a^{i,j}_{h,x_4}, a^{r,i,j}_{h,x_4}, a^{c,i,j}_{h,x_4}\}\).
			Because \(\Func{c_{\mu}}{v_7}\) can reach \(\Func{c_{\mu}}{v_8}\),
			there is some \(h' \geq h\) and some \(x_5\)
			such that
			\(\Func{c_{\mu}}{v_8} \in \Set{b^{i,j}_{h', x_5}, b^{r,i,j}_{h', x_5}, b^{c,i,j}_{h', x_5}}\).

			There is exactly one path from \(\Func{c_{\mu}}{v_6}\), \(\End{R'_{4,7}}\) and \(\End{R'_{5,7}}\) to \(\Func{c_{\mu}}{v_7}\).
			In particular, \(R_{6,7}\) is exactly the subpath of \(A^{i,j}\)
			from \(\Func{c_{\mu}}{v_6}\) to \(\Func{c_{\mu}}{v_7}\).
			Note that, by definition of \(\mu\),
			the paths \(R_{4,8}\) and \(R_{5,8}\) do not intersect \(R_{6,7}\).

			By \cref{claim:distinct-arrival-at-A-i-j},
			\(R'_{4,7}\) and \(R'_{5,7}\) are vertex disjoint.
			Hence, there are \(f < g\) and \(x_2, x_3\)
			such that
			\(e \leq f\),
			\(\End{R_{4,7}'} = a^{r,i,j}_{f,x_2}\) and
			\(\End{R_{5,7}'} = a^{r,i,j}_{g,x_3}\).

			Because \(e < h \leq h'\), we have that \(b^{r,i,j}_{h,x_4},b^{c,i,j}_{h,x_4}, b^{i,j}_{h,x_4} \in \V{R_{6,8}}\).

			By \cref{claim:distinct-arrival-at-A-i-j},
			\(R'_{4,8}\) and \(R'_{5,8}\) are vertex disjoint.
			If \(\End{R_{4,8}'}\) lies before \(\Func{c_{\mu}}{v_6}\) along \(A^{i,j}\),
			then \(b^{r,i,j}_{h,x_4}, b^{c,i,j}_{h,x_4}, b^{i,j}_{h,x_4} \in \V{R_{4,8}}\).
			However, this means that \(R_{6,9}\) must intersect \(R_{4,8}\), a contradiction.
			An analogous argument holds for \(R_{5,8}\).
			
			If \(\End{R_{4,8}'}\) lies between
			\(\Func{c_{\mu}}{v_6}\) and 
			\(\Func{c_{\mu}}{v_7}\)
			along \(A^{i,j}\),
			then it intersects \(R_{6,7}\), a contradiction.
			The same holds for \(R_{5,8}\).

			Hence, \(\End{R_{4,8}}, \End{R_{5,8}}\) lie after
			\(\Func{c_{\mu}}{v_7}\)
			along \(A^{i,j}\).
			Since \(\End{R_{4,8}} \neq \End{R_{5,8}}\),
			we have that \(h' \geq h + 2\).
			However, this implies that \(R_{7,9}\) must intersect \(R_{4,8}\) and \(R_{5,8}\), a contradiction.
		
			\Case{\(\Func{c_{\mu}}{v_6}\) lies in \(A^{i,j}\) and \(\Func{c_{\mu}}{v_7}\) lies in \(B^{i,j}\).}	

			In particular, there are \(f \leq h\) and \(y_2, y_4\)
			such that
			\(\Func{c_{\mu}}{v_6} \in \{a^{i,j}_{f,y_2}, a^{r,i,j}_{f,y_2}, a^{c,i,j}_{f,y_2}\}\) and
			\(\Func{c_{\mu}}{v_7} \in \{b^{i,j}_{h,y_4}, b^{r,i,j}_{h,y_4}, b^{c,i,j}_{h,y_4}\}\).
			Because \(\Func{c_{\mu}}{v_7}\) can reach \(\Func{c_{\mu}}{v_8}\),
			there is some \(h' \geq h\) and some \(y_5\)
			such that
			\(\Func{c_{\mu}}{v_8} \in \Set{b^{i,j}_{h', y_5}, b^{r,i,j}_{h', y_5}, b^{c,i,j}_{h', y_5}}\).
			Additionally, the path \(R_{7,8}\) is the unique \(\Func{c_{\mu}}{v_7}\)-\(\Func{c_{\mu}}{v_8}\) path in \(B^{i,j}\).

			By \cref{claim:distinct-arrival-at-A-i-j},
			for each \(t \in \{6,7,8\}\),
			\(R_{4,t}'\) and \(R'_{5,t}\) are vertex disjoint.
			Thus, there are \(e_1 \neq e_2, e_3 \neq e_4, e_5 \neq e_6\) and
			not necessarily distinct \(x_1, x_2, x_3, x_4, x_5, x_6\)
			such that
			\(\End{R_{4,6}'} = a^{r,i,j}_{e_1,x_1}\),
			\(\End{R_{5,6}'} = a^{r,i,j}_{e_2,x_2}\),
			\(\End{R_{4,7}'} = a^{r,i,j}_{e_3,x_3}\),
			\(\End{R_{5,7}'} = a^{r,i,j}_{e_4,x_4}\),
			\(\End{R_{4,8}'} = a^{r,i,j}_{e_5,x_5}\),
			\(\End{R_{5,8}'} = a^{r,i,j}_{e_6,x_6}\),
			\(e_1, e_2 \leq f\).

			Because \(\End{R_{4,7}'}, \End{R_{5,7}'}\) can reach \(\Func{c_{\mu}}{v_7}\),
			we have that \(e_3, e_4 < e_1, e_2 \leq f\).
			This implies that \(b^{r,i,j}_{h,x_4} \in \V{R_{4,7}} \cap \V{R_{5,7}}\).
			Additionally, \(R_{6,9}\) cannot intersect \(B^{i,j}\) before
			\(\Func{c_{\mu}}{v_7}\), as otherwise it would intersect
			\(R_{4,7}\).

			If \(e_5 < f\) or \(e_6 < f\),
			then one of \(R_{4,8}, R_{5,8}\)
			must intersect \(R_{7,8}\) or
			it must intersect both \(R_{4,6}\) and \(R_{5,6}\).
			This contradicts the definition of these paths.
			Hence, \(f \leq e_5, e_6\).

			We now have that \(R_{6,9}\) intersects \(R_{7,8}\) or
			it intersects one of \(R_{4,8}, R_{5,8}\).
			Both cases lead to a contradiction.
		\end{Subcase}
		\end{CaseDistinction}

		In both cases above, we obtain a contradiction.
		Thus, the initial assumption that \(D\) contains a \(\TT{9}\) as a butterfly minor is false, as desired.
\end{proof}

The analysis of \cref{stat:no-tournament} is essentially tight, as shown by the observation below.

\begin{observation}
    Let $n \geq 5$ and $k$ be a natural number and define \(S_{i, j} \coloneq [n] \times [n]\) for all \(i, j \in [k]\) as an instance of \GridTiling/.
		Let \((D, T)\) be the reduced instance constructed by \cref{const:grid-tiling-no-grid}.
    Then \(D\) contains a \(\TT{8}\) as a butterfly minor as shown in \cref{fig:K-8-in-D}.

    \begin{figure}[htbp]
        \centering
        \includegraphics[width=0.94\textwidth]{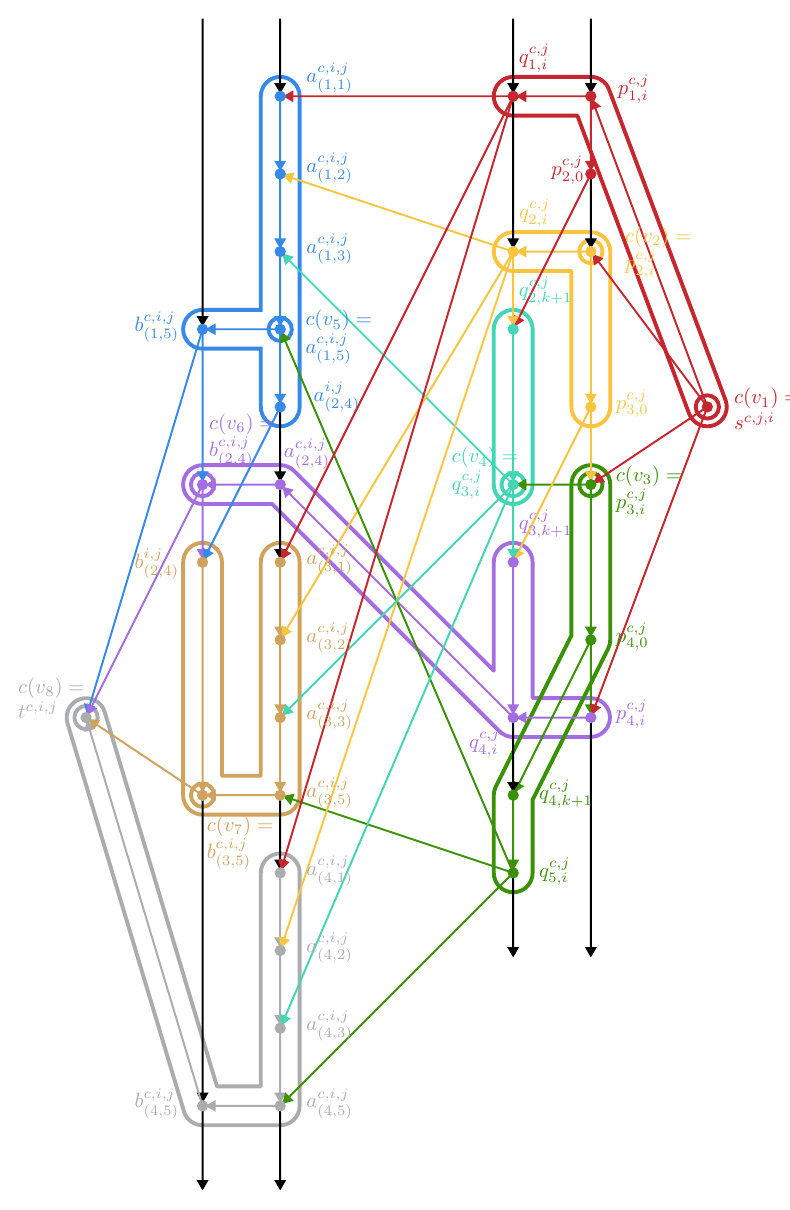}
        \caption{
            A butterfly model \(\mu\) of the \(\TT{8}\) in \(D\).
						Subdivided edges are still drawn as straight lines.
                        The left part is $\VER{i}{j}$ and the right part is $\SEL{c}{j}$.
						The subgraph \(\mu(v_i)\) is indicated by the colored region containing the center \(\Func{c_{\mu}}{v_i}\) and
						the path \(\mu((v_i, v_j))\) is indicated by the paths leaving and entering the regions
						and is colored as \(\mu(v_i)\).
                                }
        \label{fig:K-8-in-D}
    \end{figure}
\end{observation}
\subsection{Bounded ear anonymity}
\label{sec:vdp-low-ea}

\begin{lemma}
	\label{stat:low-anonymity}
	Let \((n, \{S_{i, j}\}_{i, j = 1}^k)\) be an instance of \GridTiling/.
	Let \((D, T)\) be the reduced instance constructed by \cref{const:grid-tiling-no-grid}.
	Then \(\eanon{D} \leq 5\).
\end{lemma}
\begin{proof}
	Let \(P\) be a maximal path in \(D\).
	Let \(s\) be the starting point and \(t\) the endpoint of \(P\).
	We construct an ear identify sequence \(\bar{a}\) for \(P\) as follows.
	
	Let \(\mathcal{H}\) be the set of paths of the form
		\(P^{\circ, i}, Q^{\circ, i}, A^{i, j}\) or \(B^{i, j}\) in \(D\).
	Let \(F\) be
	the set of edges 
	leaving some path in \(\mathcal{H}\) and
	entering a different path in \(\mathcal{H}\).
	The next claim follows easily from the construction of \(D\).
	\begin{claim}
		\label{stat:unique-edge-between-horizontal-paths}
		Let \((v,u) \in F\) be an edge from some path \(R_1\) to another path \(R_2\),
		where \(R_1\) is some path of the form 
		\(P^{\circ, i}\), \(Q^{\circ, i}\) or \(A^{i, j}\), and
		\(R_2\) is some path of the form
		\(Q^{\circ, i}, A^{i, j}\) or \(B^{i, j}\).
		Let \(v_1 \in \V{R_1}\) be a vertex which can reach \(v\) and
		let \(u_2 \in \V{R_2}\) be a vertex which can be reached by \(u\).
		Then there is exactly one \(v_1\)-\(u_2\) path in \(D\)
		using \((v, u)\).
	\end{claim}

	Let \(e_s\) be the outgoing edge of \(s\) in \(P\) and
	let \(e_t\) be the incoming edge of \(t\) in \(P\).
	Add \(e_s\) to \(\bar{a}\).
	Add every arc in \(F \cap \A{P}\) to \(\bar{a}\).
	Add \(e_t\) to \(\bar{a}\).
	As \(\Abs{F \cap \A{P}} \leq 3\), we have that \(\Abs{\bar{a}} \leq 5\).

	We show that \(P\) is the unique maximal path of \(D\)
	visiting \(\bar{a}\) in order.
	Observe that every edge leaving some source in \(D\)
	must end in some path in \(\mathcal{H}\).
	Similarly, every edge entering some sink in \(D\)
	must start at some path in \(\mathcal{H}\).
	Let \(R_1 \in \mathcal{H}\) be the path
	in which \(e_s\) ends and
	let \(R_2 \in \mathcal{H}\) be the path
	in which \(e_t\) starts.

	If \(R_1 = R_2\), then there is exactly one path in \(D\)
	from the head of \(e_s\) to the tail of \(e_t\). 
	Otherwise, there is an edge in \(F \cap \A{P}\).
	Let \(P'\) be the shortest subpath of \(P\)
	containing all edges in \(F \cap \A{P}\).
	By \cref{stat:unique-edge-between-horizontal-paths},
	\(P'\) is the unique minimal path in \(D\) visiting those edges
	Further, by construction there is a unique path from 
	the head of \(e_s\) to the start of \(P'\) and
	a unique path from the end of \(P'\) to the tail of \(e_t\).
	Hence, \(P\) is the unique path in \(D\) visiting the edges of \(\bar{a}\) in order.
	As \(\Abs{\bar{a}} \leq 5\), we conclude that \(\eanon{D} \leq 5\).
\end{proof}

\subsection{\W/[1]-hardness}
\label{sec:vdp-hardness}

\TaskMarcelo[done]{Proof of \cref{statement:linkage-is-w1-hard-even-with-no-grid-minor}.}
\hardnessVertexDisjoint
\begin{proof}
	We provide a reduction from \GridTiling/, which is
	\W/[1]-hard with respect to \(k\) by \cref{stat:grid-tiling-hard}.
	Let \(I_G = (n, \{S_{i, j}\}_{i, j = 1}^k)\) be a \GridTiling/ instance and
	let \(g \geq 1\) be an integer.
	We construct a \kLinkageCongestionShort/ instance \(I_L = (D, T, g)\) by using \cref{const:grid-tiling-no-grid}.
	By \cref{stat:hardness-correct},
	\(I_L\) is a \emph{yes} instance if, and only if, \(I_G\) is a \emph{yes} instance.
	By \cref{stat:no-large-grid,stat:no-tournament}, 
	\(D\) contains no \(\TT{9}\) and
	no acyclic \((5, 5)\)-grid as a butterfly minor.
	Further, \(D\) is acyclic.
	Additionally, by \cref{stat:low-anonymity}, \(\eanon{D} \leq 5\).
	Finally, \(\Abs{T} \in O(gk^2)\), and so
	the existence of an \(f(\Abs{T})n^{o(\sqrt{\Abs{T} / g})}\)-time algorithm for
	\kLinkageCongestionShort/ would imply the existence of an
	\(f(k)n^{o(k)}\) algorithm for \GridTiling/,
	contradicting \cref{stat:grid-tiling-hard}.
\end{proof}

We observe that \kLinkageCongestionShort/ can be solved in polynomial time if
\(\Anon{D} = 1\), as in this case every \(s_i\)-\(t_i\) path is disjoint from
every \(s_j\)-\(t_j\) path if \(i \neq j\), and so
taking the shortest \(s_i\)-\(t_i\) path
for every terminal pair \((s_i, t_i)\)
yields a solution if any exist.

Further, \(\Anon{\TT{k}} = \Ceiling{\frac{k - 2}{2}}\) and
the acyclic \((p,p)\)-grid has ear anonymity \(p - 1\).
Since ear anonymity is closed under taking butterfly minors on DAGs \cite{Milani24},
bounding  \(\eanon{D}\) also excludes the presence of large acyclic grids and
\(\TT{k}\) as butterfly minors.

\section{Edge disjoint paths}
\label{sec:edp}

\subsection{Hardness reduction}
\label{sec:edp-reduction}
\begin{construction}
	\label{const:vertex-to-edge-disjoint-paths}
	Let \((D', T', g)\) be an instance of \kLinkageCongestionShort/.
	We construct an \EdgeDisjointPathsCongestionShort/ instance \((D, T, g)\) iteratively as follows.
	Apply the split operation to every vertex \(v \in \V{D'}\),
	obtaining two vertices \(v_{\text{in}}\) and \(v_{\text{out}}\) in \(D\).
	
	If \(\Indeg[D']{v} \geq 3\),
	add a minimal in-tree \(T^{\text{in}}_{v}\) with \(\Indeg[D']{v}\) leaves
	where each vertex of the in-tree has indegree at most 2.
	Identify each vertex in \(\InN[D]{v_{\text{in}}}\) with
	a distinct leaf of \(T^{\text{in}}_{v}\) and
	remove the edges from \(\InN[D]{v_{\text{in}}}\) to \(v_{\text{in}}\).
	Finally, add an edge from the root of \(T^{\text{in}}_{v}\)
	to \(v_{\text{in}}\).

	If \(\Outdeg[D']{v} \geq 3\),
	we proceed in an analogous fashion as above,
	adding an out-tree \(T^{\text{out}}_{v}\) instead.
	This completes the construction of \(D\).
	Note that 
	\(\Outdegree[D]{v} + \Indegree[D]{v} \leq 3\) holds
	for all \(v \in \V{D}\).

	The multiset of terminal pairs \(T\) is obtained from \(T'\)
	by replacing every occurrence of \((s,t) \in T'\) with
	\((s_{\text{in}}, t_{\text{out}})\).
\end{construction}

\subsection{Excluding a \texorpdfstring{\((7,7)\)}{(7, 7)}-wall}
\label{sec:edp-no-wall}
It is not too difficult to see that the digraph \(D\) given by
\cref{const:vertex-to-edge-disjoint-paths}
contains an acyclic \((5, 5)\)-wall as an immersion.
To exclude the acyclic \((7, 7)\)-wall,
we use the fact that the original instance \(D'\)	
does not contain an acyclic \((5, 5)\)-grid.

\TaskMarcelo{Double check proof of \cref{stat:no-wall-immersion}}
\begin{lemma}
	\label{stat:no-wall-immersion}
 	Let \((D', T', g)\) be a \kLinkageCongestionShort/ instance.
 	Let \((D, T, g)\) be the \EdgeDisjointPathsCongestionShort/ instance given by
 	\cref{const:vertex-to-edge-disjoint-paths}.
	If \(D\) contains an acyclic \((7,7)\)-wall as a weak immersion,
	then \(D'\) contains an acyclic \((5,5)\)-grid as a butterfly minor.
\end{lemma}
	
\begin{proof}
	Assume that \(D\) contains an acyclic \((7,7)\)-wall \(W\) as a weak immersion.
		
	Let \(\mu\) be a weak immersion of \(W\) in \(D\).
	For each \(v \in \V{D'}\), let \(Z_v = \V{T^{\text{in}}_{v}} \cup \V{T^{\text{out}}_{v}} \cup \{v_{\text{in}}, v_{\text{out}}\}\).
	\Cref{claim:degree-zone} is immediate from the degrees of the vertices in \(D\).
	\begin{claim}
		\label{claim:degree-zone}
		For every \(u \in \V{W}\) with \(\Outdegree[W]{u} = 2\)
		there is some \(v \in \V{D'}\)
		such that
		\(\mu(u) \in \{v_{\text{out}}\} \cup \V{T^{\text{out}}_{v}}\).
		Similarly, if \(\Indegree[W]{u} = 2\),
		we have
		\(\mu(u) \in \{v_{\text{in}}\} \cup \V{T^{\text{in}}_{v}}\)
		for some \(v\).
	\end{claim}

	\begin{claim}
		\label{claim:path-in-zone}
		If \(\mu(u_1), \mu(u_2) \in Z_v\),
		then \(\mu(u_3) \in Z_v\) for every \(\mu(u_3) \in \V{W}\)
		which lies on a path from \(u_1\) to \(u_2\).
	\end{claim}
	\begin{claimproof}
			Without loss of generality, assume that \(u_1\) can reach \(u_2\) in \(W\).
			Let \(u_3\) be a vertex in a \(u_1\)-\(u_2\) path in \(W\) and
			let \(x \in \V{D'}\) 
			such that
			\(\mu(u3) \in Z_x\).
			By definition of immersion,
			there is a \(\mu(u_1)\)-\(\mu(u_3)\) path and
			a \(\mu(u_3)\)-\(\mu(u_2)\) path in \(D\).
			By construction of \(D\), this is only possible
			if there is a \(\V{Z_x}\)-\(\V{Z_v}\) and
			a \(\V{Z_v}\)-\(\V{Z_x}\) path in \(D\).
			As \(D\) and \(D'\) are acyclic,
			this  implies that \(x = v\), as desired.
	\end{claimproof}

	Let \(V'\) be the set of split vertices of \(W\),
	that is, the vertices of the acyclic \((7,7)\)-grid
	that were split in order to construct \(W\) (see \cref{fig:7x7-wall}).

	\begin{figure}[!h]
			\centering
			\begin{tikzpicture}[xscale=0.6, yscale=0.75]
				\input{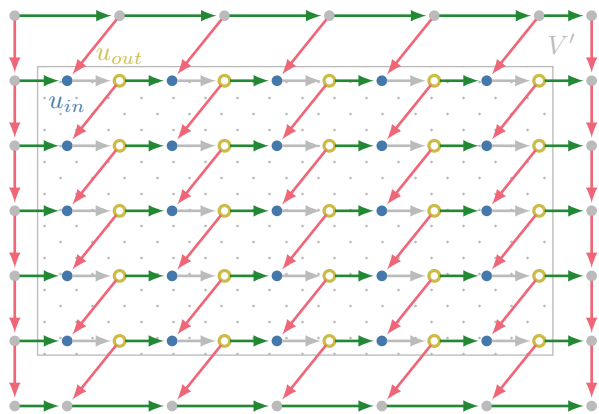}
			\end{tikzpicture}
			\caption{An acyclic \((7,7)\)-wall.
			The filled blue vertices correspond to
			the \emph{in} vertex of a split vertex in the acyclic \((7,7)\)-grid,
			while the hollow yellow vertices correspond to
			the \emph{out} vertex.}
			\label{fig:7x7-wall}
	\end{figure}

	Let \(w_1, w_2 \in V'\) be two distinct vertices.
	Let \(v_1, v_2 \in \V{D'}\) be two vertices
	such that
	\(\mu(w_1) \in Z_{v_1}\) and \(\mu(w_2) \in Z_{v_2}\).
	We claim that \(v_1 \neq v_2\).
	Assume towards a contradiction that this is not the case.

	By \cref{claim:degree-zone}, \(v_1 \neq v_2\) if \(w_1\) and \(w_2\) have different in- or outdegree.
	Hence, we assume, without loss of generality, that \(\Outdegree{w_1} = \Outdegree{w_2} = 2\).
	The case where \(\Indegree{w_1} = \Indegree{w_2} = 2\) follows analogously.
	If there is a path from \(w_1\) to \(w_2\) in \(W\),
	then this path must contain some \(w_3 \in \V{W}\)
	such that \(\Indegree[W]{w} = 2\).
	Hence, \(\mu(w_3) \not \in Z_{v_1} \cup Z_{v_2}\).
	By \cref{claim:path-in-zone}, we obtain that \(v_1 \neq v_2\), as desired.

	So assume that \(w_1\) and \(w_2\) are incomparable in the topological ordering.
	Let \(w_3 \in V'\) be the unique vertex which can reach all vertices of \(V'\) in \(W\)
	(for the case where \(\Indegree{w_1} = \Indegree{w_2} = 2\),
	we pick the unique vertex which can be reached by every vertex in \(V'\) instead).
	By definition of immersion,
	there are two edge-disjoint paths
	from \(\mu(w_3)\) to \(\mu(w_1)\) and \(\mu(w_2)\) in \(D\).
	By \cref{claim:degree-zone}, these paths can only exist if
	\(\mu(w_3) \in Z_{v_1} = Z_{v_2}\) as well.

	As \(\Outdegree[W]{w_3} = 2\), we know that \(w_3, w_1, w_2\) are pairwise non-neighbors in \(W\).
	This implies that
	the paths from \(w_3\) to \(w_1\) and \(w_2\)
	contain some vertex \(w_4\)
	where \(\Indegree[W]{w_4} = 2\).
	However, by \cref{claim:path-in-zone}, 
	\(w_4 \in Z_{v_1}\), a contradiction to \cref{claim:degree-zone}.
	
	Hence, we conclude that \(v_1 \neq v_2\).

	We construct a butterfly model \(\phi\) of an acyclic \((5,5)\)-grid \(G\) in \(D'\) as follows.
	To simplify notation, 
	we consider \(G\) to be the inner \((5, 5)\)-subgrid
	of the acyclic \((7,7)\)-grid
	used to construct the acyclic \((7,7)\)-wall \(W\).
	For each \(u \in \V{G}\),
	let \(u_{\text{in}}, u_{\text{out}} \in V'\)
	be the pair of vertices into which \(u\) was split when constructing the wall.
	Let \(P_u\) be the \(\mu(u_{\text{in}})\)-\(\mu(u_{\text{out}})\) path in \(D\)
	given by \(\mu((u_{\text{in}}, u_{\text{out}}))\).
	For each \((u, w) \in \A{G}\),
	let \(Q_{u,w}\) be the
	\(\mu(u_{\text{out}})\)-\(\mu(w_{\text{in}})\) path in \(D\)
	given by \(\mu((u_{\text{out}}, w_{\text{in}}))\).

	We construct a path \(P'_u\) and a path \(Q'_{u,w}\) in \(D'\) 
	by replacing every occurrence of an edge
	\((v_{\text{in}}, v_{\text{out}})\) in \(P_u\) or \(Q_{u,w}\)
	with the vertex \(v \in \V{D'}\).
	By construction of \(D\), these paths are indeed paths in \(D'\).

	Set \(\phi(u) = P'_u\) and \(\phi((u,w)) = Q'_{u,w}\).
	Observe that \(\Func{c_{\mu}}{u} = \mu(u_{\text{in}})\).
	Because each vertex in \(D\) has outdegree plus indegree at most 3,
	we know that edge-disjoint paths in \(D\) are necessarily also vertex disjoint.
	Hence, all paths \(P_{u}\) and \(Q_{u,w}\) defined above are also pairwise internally vertex disjoint,
	which in turn implies that all \(P'_{u}\) and \(Q'_{u,w}\) are pairwise internally vertex disjoint as well.
	Thus, for each \((u_1, u_2) \in \V{G}\), 
	there is a \(\phi(u_1)\)-\(\phi(u_2)\) path in \(D'\) and
	these paths are internally vertex disjoint.

	Hence, \(D'\) contains the \((5, 5)\)-acyclic grid \(G\) as a butterfly minor.
\end{proof}

\subsection{Bounded ear anonymity}
\label{sec:edp-low-ea}

\TaskMarcelo[done]{Proof of \cref{stat:edge-disjoint-construction-path-mapping}.}

\begin{observation}
	\label{stat:edge-disjoint-construction-path-mapping}
 	Let \((D', T', g)\) be a \kLinkageCongestionShort/ instance where \(D'\) is a DAG.
 	Let \((D, T, g)\) be the \EdgeDisjointPathsCongestionShort/ instance given by
 	\cref{const:vertex-to-edge-disjoint-paths}.
	Let \(P\) be a maximal path in \(D\) and
	let \(P_{1} \cdot P_{2} \cdot \ldots \cdot P_{n} = P\) be a partition of \(P\)
	into subpaths
	such that 
	\(P_i\) is a maximal subpath of \(P\) contained inside
	\(T_{v_i}^{\text{in}}, (v_{i, \text{in}}, v_{i, \text{out}})\) and \(T_{v_i}^{\text{out}}\)
	for some \(v_i \in \V{D'}\).
	Let \(P'\) be the path obtained by replacing every \(P_i\) in \(P\) by \(v_i\).
	Then \(P'\) is a path in \(D'\).
\end{observation}
	
\begin{proof}
	Consider a subpath \(P_i\) of \(P\) and the corresponding vertex \(v_i \in \V{D'}\) as described above.
	After visiting some vertex in \(T_{v_i}^{\text{in}}\),
	\(P_i\) must contain the edge \((v_{i, \text{in}}, v_{i, \text{out}}) \in \A{D}\)
	as \(T_{v_i}^{\text{in}}\) is an in-tree.
	Analogously, before reaching some vertex in \(T_{v_i}^{\text{out}}\),
	the path \(P_i\) must first visit \((v_{i, \text{in}}, v_{i, \text{out}}) \in \A{D}\).

	Since \(D\) is a DAG, we have that \(P\) starts in a source and ends in a sink of \(D\).
	Hence, \(v_i \neq v_j\) if \(i \neq j\).
	Further, if there is an edge from \(T_{v_i}^{\text{out}}\) to \(T_{v_j}^{\text{in}}\)
	in \(D\), then \((v_i, v_j) \in \A{D'}\).
	Thus, the sequence \(P'\) given by \((v_{1}, v_{2}, \ldots, v_{n})\) is a path in \(D'\), as desired.
\end{proof}
	
\TaskMarcelo[done]{Proof of \cref{stat:edge-low-anonymity}.}
\begin{lemma}
	\label{stat:edge-low-anonymity}
 	Let \((D', T', g)\) be a \kLinkageCongestionShort/ instance.
 	Let \((D, T, g)\) be the \EdgeDisjointPathsCongestionShort/ instance given by
 	\cref{const:vertex-to-edge-disjoint-paths}.
	Then \(\eanon{D} \leq \eanon{D'}\).
\end{lemma}
	
\begin{proof}
	Let \(P\) be a maximal path in \(D\).
	Clearly, \(P\) starts in some \(s_{\text{in}}\) and
	ends in some \(t_{\text{out}}\) for \(s, t \in \V{D'}\).
	Let \(P'\) be the corresponding path in \(D'\) given by
	\cref{stat:edge-disjoint-construction-path-mapping}.
	Let \(\bar{b}\) be an ear-identifying sequence for \(P'\) in \(D'\).
	Construct an ear-identifying sequence for \(P\) in \(D\) as follows.

	For each edge \((v,u)\) in \(\bar{b}\),
	add the edge \((v', u')\) to \(\bar{a}\),
	where \(v'\) is the leaf of \(T_v^{\text{out}}\)
	which contains an edge to	\(T_u^{\text{in}}\) and
	where \(u'\) is the leaf of \(T_u^{\text{in}}\)
	which contains an edge from \(T_v^{\text{out}}\).
	Because \((v,u)\) is an edge in \(D'\),
	the vertices \(v', u'\) exist and
	have an edge between them by construction of \(D\).

	Let \(Q\) be a maximal path in \(D\) visiting the edges of \(\bar{a}\) in order.
	Let \(Q'\) be the corresponding path in \(D'\) given by \cref{stat:edge-disjoint-construction-path-mapping}.
	Clearly, \(Q'\) visits all edges of \(\bar{b}\).
	Hence, \(Q' = P'\), as \(\bar{b}\) is an ear-identifying sequence for \(P'\).

	For each \(v_i\) in \(Q'\),
	let \(D_i\) be the induced subgraph of \(D\)
	given by
	\(T_{v_i}^{\text{in}}, (v_{i, \text{in}}, v_{i, \text{out}})\) and \(T_{v_i}^{\text{out}}\).
	Let \(Q_i\) be the maximal subpath of \(Q\) inside \(D_i\) and
	let \(P_i\) be the maximal subpath of \(P\) inside \(D_i\).

	Because \(T_{v_i}^{\text{in}}\) is an in-tree,
	there is exactly one path from a leaf of \(T_{v_i}^{\text{in}}\) to \(v_{i, \text{in}}\).
	Similarly, there is exactly one path from \(v_{i, \text{out}}\).
	As the edges of \(\bar{a}\) connect the leaves of such trees,
	we have \(P_i = Q_i\) for all \(v_i \in Q'\).
	Hence, \(P = Q\), and so \(\bar{a}\) is an ear-identifying sequence for \(P\).
\end{proof}

\TaskMarcelo[done]{Add results regarding the ETH.}

\subsection{\W/[1]-hardness}
\label{sec:edp-hardness}
The last remaining step for our hardness result is to show that the reduction given by
\cref{const:vertex-to-edge-disjoint-paths} is correct.

\TaskMarcelo[done]{Proof of \cref{statement:linkage-is-w1-hard-even-with-no-wall-immersion}}
\hardnessEdgeDisjoint

\begin{proof}
	We provide a reduction from \kLinkageCongestionShort/,
	which is \W/[1]-hard with respect to \(k\) by \cref{statement:linkage-is-w1-hard-even-with-no-grid-minor},
	even if the input digraph \(D'\) is acyclic, does not contain a \((5, 5)\)-acyclic grid and
	\(\eanon{D'} \leq 5\).
 	
	Let \(I_V = (D', T', g)\) be a \kLinkageCongestionShort/ instance.
 	Let \(I_E = (D, T, g)\) be the \EdgeDisjointPathsCongestionShort/ instance given by
 	\cref{const:vertex-to-edge-disjoint-paths}.
	
	By \cref{stat:no-wall-immersion}, \(D\) does not contain an acyclic \((7,7)\)-grid.
	By \cref{stat:edge-low-anonymity}, \(\eanon{D} \leq \eanon{D'} \leq 5\).
	It is immediate from the construction that all vertices have undirected degree at most 3.
	Additionally, as \cref{const:vertex-to-edge-disjoint-paths} replaces vertices in \(D'\)
	by acyclic digraphs, we have that \(D\) is acyclic if, and only if, \(D'\) is acyclic.
	Finally, \(\Abs{T} = \Abs{T'}\), and so an algorithm with running time 
	\(f(k)n^{o(\sqrt{k/g})}\) for \EdgeDisjointPathsCongestionShort/
	would imply an algorithm with the same running time for \kLinkageCongestionShort/,
	contradicting the ETH by \cref{statement:linkage-is-w1-hard-even-with-no-grid-minor}.

	We show that \(I_V\) is a yes-instance if, and only if,
	\(I_E\) is a yes-instance.
	Recall the definition of \(T_v^{\text{in}}, T_v^{\text{out}}, v_{\text{in}}\) and \(v_{\text{out}}\)
	from \cref{const:vertex-to-edge-disjoint-paths}.

  \textbf{First direction:} If \(I_V\) is a \emph{yes} instance, then \(I_E\) is a \emph{yes} instance.

	Let \(\mathcal{L}'\) be a solution for \(I_V\).
	Let \((s_i, t_i) \in T'\) and let \(L_i' \in \mathcal{L}'\) be the corresponding
	\(s_i\)-\(t_i\) path in \(\mathcal{L}'\).
	
	Construct a path \(L_i\) from \(s_{i, \text{in}}\) to \(t_{i, \text{out}}\) in \(D'\) as follows.
	Let \(v \in \V{L_i'}\).
	If \(v = s_i\), let \(P_v^{\text{in}}\) be the empty path.
	Similarly, if \(v = t_i\), let \(P_v^{\text{out}}\) be the empty path.

	If \(v\) has a predecessor \(u\) along \(L_i'\),
	let \(u'\) be the unique leaf of \(T_v^{\text{in}}\)
	which as an incoming edge from \(T_{u}^{\text{out}}\).
	Define \(P_v^{\text{in}}\) as the unique path from \(u'\) to \(v_{\text{in}}\).

	Similarly, if \(v\) has a successor \(u\) along \(L_i'\),
	let \(u'\) be the unique leaf of \(T_v^{\text{out}}\)
	which has an outgoing edge to \(T_{u}^{\text{in}}\).
	Define \(P_v^{\text{out}}\) as the unique path from \(v_{\text{out}}\) to \(u'\).

	Construct \(L_i\) by
	replacing every vertex \(v\) in \(L_i'\) with the path
	\(P_v^{\text{in}} \cdot (v_{\text{in}}, v_{\text{out}}) \cdot P_v^{\text{out}}\).
	Clearly, \(L_i\) is a path from \(s_{i, \text{in}}\) to \(t_{i, \text{out}}\).
	Further, each edge \((v_{\text{in}}, v_{\text{out}})\)
	is used by at most \(g\) paths as \(\mathcal{L}'\) is a set of paths with vertex congestion at most \(g\).
		Further, each edge in \(T_{v}^{\text{in}}\) and in \(T_{v}^{\text{out}}\) is used at most once per
	path in \(\mathcal{L}'\) containing \(v\).
	Hence, the congestion of the paths \(L_i\) constructed above is at most \(g\), as desired.

  \textbf{Second direction:} If \(I_E\) is a \emph{yes} instance, then \(I_V\) is a \emph{yes} instance.

	Let \(\mathcal{L}\) be a solution for \(I_E\).
	Let \((s_i, t_i) \in T'\) and let \(L_i \in \mathcal{L}\) be the corresponding
	\(s_{i, \text{in}}\)-\(t_{i, \text{out}}\) path in \(\mathcal{L}\).

	Using \(L_i\), construct a path \(L'_i\) from \(s_i\) to \(t_i\)
	as described in \cref{stat:edge-disjoint-construction-path-mapping}.
	
	Each \(T_{v}^{\text{in}}\) and \(T_{v}^{\text{out}}\)
	can be used by at most \(g\) paths in \(\mathcal{L}\),
	as the edge \((v_{\text{in}}, v_{\text{out}})\)
	separates \(T_{v}^{\text{in}}\) from all sinks
	and \(T_{v}^{\text{out}}\) from all sources in \(D\) and
	can be used at most \(g\) times by \(\mathcal{L}\).
	Hence, the paths \(L_i'\) have congestion at most \(g\) and
	connect \(s_i\) to \(t_i\) in \(D'\), as desired.
\end{proof}

\bibliography{references}

\begin{thebibliography}{10}

\bibitem{AKMR19}
Saeed~Akhoondian Amiri, Stephan Kreutzer, D{\'{a}}niel Marx, and Roman
  Rabinovich.
\newblock Routing with congestion in acyclic digraphs.
\newblock {\em Inf. Process. Lett.}, 151, 2019.
\newblock URL: \url{https://doi.org/10.1016/j.ipl.2019.105836}, \href
  {https://doi.org/10.1016/J.IPL.2019.105836}
  {\path{doi:10.1016/J.IPL.2019.105836}}.

\bibitem{BG2009}
J{\o}rgen Bang{-}Jensen and Gregory~Z. Gutin.
\newblock {\em Digraphs - Theory, Algorithms and Applications, Second Edition}.
\newblock Springer Monographs in Mathematics. Springer, 2009.

\bibitem{CKK24}
Dario~Giuliano Cavallaro, Ken{-}ichi Kawarabayashi, and Stephan Kreutzer.
\newblock Edge-disjoint paths in eulerian digraphs.
\newblock In Bojan Mohar, Igor Shinkar, and Ryan O'Donnell, editors, {\em
  Proceedings of the 56th Annual {ACM} Symposium on Theory of Computing, {STOC}
  2024, Vancouver, BC, Canada, June 24-28, 2024}, pages 704--715. {ACM}, 2024.
\newblock \href {https://doi.org/10.1145/3618260.3649758}
  {\path{doi:10.1145/3618260.3649758}}.

\bibitem{ChalermsookLN14}
Parinya Chalermsook, Bundit Laekhanukit, and Danupon Nanongkai.
\newblock Pre-reduction graph products: Hardnesses of properly learning dfas
  and approximating {EDP} on dags.
\newblock In {\em 55th {IEEE} Annual Symposium on Foundations of Computer
  Science, {FOCS} 2014, Philadelphia, PA, USA, October 18-21, 2014}, pages
  444--453. {IEEE} Computer Society, 2014.
\newblock \href {https://doi.org/10.1109/FOCS.2014.54}
  {\path{doi:10.1109/FOCS.2014.54}}.

\bibitem{ChekuriKS04}
Chandra Chekuri, Sanjeev Khanna, and F.~Bruce Shepherd.
\newblock Edge-disjoint paths in planar graphs.
\newblock In {\em 45th Symposium on Foundations of Computer Science {(FOCS}
  2004), 17-19 October 2004, Rome, Italy, Proceedings}, pages 71--80. {IEEE}
  Computer Society, 2004.
\newblock \href {https://doi.org/10.1109/FOCS.2004.27}
  {\path{doi:10.1109/FOCS.2004.27}}.

\bibitem{ChekuriKS06}
Chandra Chekuri, Sanjeev Khanna, and F.~Bruce Shepherd.
\newblock An {O}(sqrt(n)) approximation and integrality gap for disjoint paths
  and unsplittable flow.
\newblock {\em Theory Comput.}, 2(7):137--146, 2006.
\newblock URL: \url{https://doi.org/10.4086/toc.2006.v002a007}, \href
  {https://doi.org/10.4086/TOC.2006.V002A007}
  {\path{doi:10.4086/TOC.2006.V002A007}}.

\bibitem{ChekuriKS09}
Chandra Chekuri, Sanjeev Khanna, and F.~Bruce Shepherd.
\newblock Edge-disjoint paths in planar graphs with constant congestion.
\newblock {\em {SIAM} J. Comput.}, 39(1):281--301, 2009.
\newblock \href {https://doi.org/10.1137/060674442}
  {\path{doi:10.1137/060674442}}.

\bibitem{Chitnis23}
Rajesh Chitnis.
\newblock A tight lower bound for edge-disjoint paths on planar dags.
\newblock {\em {SIAM} J. Discret. Math.}, 37(2):556--572, 2023.
\newblock URL: \url{https://doi.org/10.1137/21m1395089}, \href
  {https://doi.org/10.1137/21M1395089} {\path{doi:10.1137/21M1395089}}.

\bibitem{ChuzhoyKL16}
Julia Chuzhoy, David H.~K. Kim, and Shi Li.
\newblock Improved approximation for node-disjoint paths in planar graphs.
\newblock In Daniel Wichs and Yishay Mansour, editors, {\em Proceedings of the
  48th Annual {ACM} {SIGACT} Symposium on Theory of Computing, {STOC} 2016,
  Cambridge, MA, USA, June 18-21, 2016}, pages 556--569. {ACM}, 2016.
\newblock \href {https://doi.org/10.1145/2897518.2897538}
  {\path{doi:10.1145/2897518.2897538}}.

\bibitem{ChuzhoyKN22}
Julia Chuzhoy, David H.~K. Kim, and Rachit Nimavat.
\newblock New hardness results for routing on disjoint paths.
\newblock {\em {SIAM} J. Comput.}, 51(2):17--189, 2022.
\newblock URL: \url{https://doi.org/10.1137/17m1146580}, \href
  {https://doi.org/10.1137/17M1146580} {\path{doi:10.1137/17M1146580}}.

\bibitem{CFKLMPPS15}
Marek Cygan, Fedor~V. Fomin, Lukasz Kowalik, Daniel Lokshtanov, D{\'{a}}niel
  Marx, Marcin Pilipczuk, Michal Pilipczuk, and Saket Saurabh.
\newblock {\em Parameterized Algorithms}.
\newblock Springer, 2015.
\newblock \href {https://doi.org/10.1007/978-3-319-21275-3}
  {\path{doi:10.1007/978-3-319-21275-3}}.

\bibitem{CyganMPP13}
Marek Cygan, D{\'{a}}niel Marx, Marcin Pilipczuk, and Michal Pilipczuk.
\newblock The planar directed k-vertex-disjoint paths problem is
  fixed-parameter tractable.
\newblock In {\em 54th Annual {IEEE} Symposium on Foundations of Computer
  Science, {FOCS} 2013, 26-29 October, 2013, Berkeley, CA, {USA}}, pages
  197--206. {IEEE} Computer Society, 2013.
\newblock \href {https://doi.org/10.1109/FOCS.2013.29}
  {\path{doi:10.1109/FOCS.2013.29}}.

\bibitem{DowneyF13}
Rodney~G. Downey and Michael~R. Fellows.
\newblock {\em Fundamentals of Parameterized Complexity}.
\newblock Texts in Computer Science. Springer, 2013.
\newblock \href {https://doi.org/10.1007/978-1-4471-5559-1}
  {\path{doi:10.1007/978-1-4471-5559-1}}.

\bibitem{EneMPR16}
Alina Ene, Matthias Mnich, Marcin Pilipczuk, and Andrej Risteski.
\newblock On routing disjoint paths in bounded treewidth graphs.
\newblock In Rasmus Pagh, editor, {\em 15th Scandinavian Symposium and
  Workshops on Algorithm Theory, {SWAT} 2016, June 22-24, 2016, Reykjavik,
  Iceland}, volume~53 of {\em LIPIcs}, pages 15:1--15:15. Schloss Dagstuhl -
  Leibniz-Zentrum f{\"{u}}r Informatik, 2016.
\newblock URL: \url{https://doi.org/10.4230/LIPIcs.SWAT.2016.15}, \href
  {https://doi.org/10.4230/LIPICS.SWAT.2016.15}
  {\path{doi:10.4230/LIPICS.SWAT.2016.15}}.

\bibitem{FHW1978}
Steven Fortune, John Hopcroft, and James Wyllie.
\newblock The directed subgraph homeomorphism problem.
\newblock {\em Theoretical Computer Science}, 10(2):111--121, 1980.
\newblock \href
  {https://doi.org/http://dx.doi.org/10.1016/0304-3975(80)90009-2}
  {\path{doi:http://dx.doi.org/10.1016/0304-3975(80)90009-2}}.

\bibitem{ImpagliazzoP2001}
Russell Impagliazzo and Ramamohan Paturi.
\newblock On the complexity of k-sat.
\newblock {\em J. Comput. Syst. Sci.}, 62(2):367--375, 2001.
\newblock URL: \url{https://doi.org/10.1006/jcss.2000.1727}, \href
  {https://doi.org/10.1006/JCSS.2000.1727} {\path{doi:10.1006/JCSS.2000.1727}}.

\bibitem{ImpagliazzoPZ2001}
Russell Impagliazzo, Ramamohan Paturi, and Francis Zane.
\newblock Which problems have strongly exponential complexity?
\newblock {\em J. Comput. Syst. Sci.}, 63(4):512--530, 2001.
\newblock URL: \url{https://doi.org/10.1006/jcss.2001.1774}, \href
  {https://doi.org/10.1006/JCSS.2001.1774} {\path{doi:10.1006/JCSS.2001.1774}}.

\bibitem{johnson2001directed}
Thor Johnson, Neil Robertson, Paul Seymour, and Robin Thomas.
\newblock Directed tree-width.
\newblock {\em Journal of Combinatorial Theory, Series B}, 82(1):138--154,
  2001.

\bibitem{KawarabayashiKR12}
Ken{-}ichi Kawarabayashi, Yusuke Kobayashi, and Bruce~A. Reed.
\newblock The disjoint paths problem in quadratic time.
\newblock {\em J. Comb. Theory {B}}, 102(2):424--435, 2012.
\newblock URL: \url{https://doi.org/10.1016/j.jctb.2011.07.004}, \href
  {https://doi.org/10.1016/J.JCTB.2011.07.004}
  {\path{doi:10.1016/J.JCTB.2011.07.004}}.

\bibitem{KleinbergT98}
Jon~M. Kleinberg and {\'{E}}va Tardos.
\newblock Approximations for the disjoint paths problem in high-diameter planar
  networks.
\newblock {\em J. Comput. Syst. Sci.}, 57(1):61--73, 1998.
\newblock URL: \url{https://doi.org/10.1006/jcss.1998.1579}, \href
  {https://doi.org/10.1006/JCSS.1998.1579} {\path{doi:10.1006/JCSS.1998.1579}}.

\bibitem{KolliopoulosS04}
Stavros~G. Kolliopoulos and Clifford Stein.
\newblock Approximating disjoint-path problems using packing integer programs.
\newblock {\em Math. Program.}, 99(1):63--87, 2004.
\newblock URL: \url{https://doi.org/10.1007/s10107-002-0370-6}, \href
  {https://doi.org/10.1007/S10107-002-0370-6}
  {\path{doi:10.1007/S10107-002-0370-6}}.

\bibitem{KorhonenPS24}
Tuukka Korhonen, Michal Pilipczuk, and Giannos Stamoulis.
\newblock Minor containment and disjoint paths in almost-linear time.
\newblock In {\em 65th {IEEE} Annual Symposium on Foundations of Computer
  Science, {FOCS} 2024, Chicago, IL, USA, October 27-30, 2024}, pages 53--61.
  {IEEE}, 2024.
\newblock \href {https://doi.org/10.1109/FOCS61266.2024.00014}
  {\path{doi:10.1109/FOCS61266.2024.00014}}.

\bibitem{Marx12}
D{\'{a}}niel Marx.
\newblock A tight lower bound for planar multiway cut with fixed number of
  terminals.
\newblock In Artur Czumaj, Kurt Mehlhorn, Andrew~M. Pitts, and Roger
  Wattenhofer, editors, {\em Automata, Languages, and Programming - 39th
  International Colloquium, {ICALP} 2012, Warwick, UK, July 9-13, 2012,
  Proceedings, Part {I}}, volume 7391 of {\em Lecture Notes in Computer
  Science}, pages 677--688. Springer, 2012.
\newblock \href {https://doi.org/10.1007/978-3-642-31594-7\_57}
  {\path{doi:10.1007/978-3-642-31594-7\_57}}.

\bibitem{Milani24}
Marcelo~Garlet Milani.
\newblock Directed ear anonymity.
\newblock In Jos{\'{e}}~A. Soto and Andreas Wiese, editors, {\em {LATIN} 2024:
  Theoretical Informatics - 16th Latin American Symposium, Puerto Varas, Chile,
  March 18-22, 2024, Proceedings, Part {II}}, volume 14579 of {\em Lecture
  Notes in Computer Science}, pages 77--97. Springer, 2024.
\newblock \href {https://doi.org/10.1007/978-3-031-55601-2\_6}
  {\path{doi:10.1007/978-3-031-55601-2\_6}}.

\bibitem{RobertsonS95b}
Neil Robertson and Paul~D. Seymour.
\newblock Graph minors .xiii. the disjoint paths problem.
\newblock {\em J. Comb. Theory {B}}, 63(1):65--110, 1995.
\newblock URL: \url{https://doi.org/10.1006/jctb.1995.1006}, \href
  {https://doi.org/10.1006/JCTB.1995.1006} {\path{doi:10.1006/JCTB.1995.1006}}.

\bibitem{slivkins2010parameterized}
Aleksandrs Slivkins.
\newblock Parameterized tractability of edge-disjoint paths on directed acyclic
  graphs.
\newblock {\em SIAM Journal on Discrete Mathematics}, 24(1):146--157, 2010.

\bibitem{Wlodarczyk2024}
Michal Wlodarczyk.
\newblock Constant approximating disjoint paths on acyclic digraphs is
  w[1]-hard.
\newblock In Juli{\'{a}}n Mestre and Anthony Wirth, editors, {\em 35th
  International Symposium on Algorithms and Computation, {ISAAC} 2024, December
  8-11, 2024, Sydney, Australia}, volume 322 of {\em LIPIcs}, pages
  57:1--57:16. Schloss Dagstuhl - Leibniz-Zentrum f{\"{u}}r Informatik, 2024.
\newblock URL: \url{https://doi.org/10.4230/LIPIcs.ISAAC.2024.57}, \href
  {https://doi.org/10.4230/LIPICS.ISAAC.2024.57}
  {\path{doi:10.4230/LIPICS.ISAAC.2024.57}}.

\end{thebibliography}

\end{document}